\documentclass[letterpaper,11pt]{article}
\usepackage[margin=1in]{geometry}
\usepackage[T1]{fontenc}
\usepackage{lmodern}
\usepackage{microtype}
\usepackage{amsmath,amsfonts,amssymb,amsthm}
\usepackage{mathtools}
\usepackage{booktabs,array}
\usepackage{enumitem}
\usepackage[table]{xcolor}
\usepackage{graphicx}
\usepackage{algorithm,algpseudocode}
\usepackage{xspace}
\usepackage{tikz}
\usepackage{thm-restate}
\usepackage{multirow,multicol}
\usepackage{physics}
\usepackage[authoryear,square,sort]{natbib}
\let\cite\citep
\usepackage{xurl}
\usepackage[hidelinks,bookmarksnumbered=true]{hyperref}
\usepackage[capitalize]{cleveref}

\makeatletter
\AtBeginDocument{\def\theHALG@line{\thealgorithm.\arabic{ALG@line}}}
\makeatother

\hypersetup{
  pdftitle={On the Theoretical Foundations of Data Exchange Economies},
  pdfauthor={Hannaneh Akrami; Bhaskar Ray Chaudhury; Jugal Garg; Aniket Murhekar},
  pdfkeywords={Data exchange, fairness, reciprocity, core-stability}
}

\newtheorem{lemma}{Lemma}
\newtheorem{theorem}{Theorem}
\newtheorem{observation}{Observation}
\newtheorem{corollary}{Corollary}
\newtheorem{definition}{Definition}
\newtheorem{claim}{Claim}
\newtheorem*{claim*}{Claim}

\newtheorem{proposition}{Proposition}
\newtheorem{remark}{Remark}
\crefname{observation}{Observation}{Observations}
\Crefname{observation}{Observation}{Observations}

\let\eps\varepsilon

\DeclareMathOperator*{\argmin}{argmin}
\DeclareMathOperator*{\argmax}{argmax}
\newcommand{\I}{\mathcal{I}}
\newcommand{\D}{\mathcal{D}}
\newcommand{\Val}{\mathcal{U}}
\newcommand{\N}{\mathcal{N}}

\newcommand{\poly}{\text{poly}}

\newcommand{\Fair}{Reciprocal}
\newcommand{\val}{u}
\newcommand{\cs}{core-stable\xspace}
\newcommand{\alloc}{\mathbf{x}}
\newcommand{\x}{\mathbf{x}}
\newcommand{\y}{\mathbf{y}}
\newcommand{\z}{\mathbf{z}}
\newcommand{\e}{\mathbf{e}}

\newcommand{\elltwo}[1]{{\lVert#1\rVert}_2}

\title{On the Theoretical Foundations of Data Exchange Economies\thanks{Work supported by NSF Grants CCF-1942321 and CCF-2334461.}}
\author{
  \begin{tabular}{c@{\hspace{3em}}c}
    Hannaneh Akrami\thanks{University of Bonn, Germany.} &
    Bhaskar Ray Chaudhury\thanks{University of Illinois at Urbana-Champaign, USA.}\\
    \href{mailto:hakrami@mpi-inf.mpg.de}{\texttt{\small hakrami@mpi-inf.mpg.de}} &
    \href{mailto:braycha@illinois.edu}{\texttt{\small braycha@illinois.edu}}\\[1.5ex]
    Jugal Garg\footnotemark[3] &
    Aniket Murhekar\thanks{Northwestern University, USA.}\\
    \href{mailto:jugal@illinois.edu}{\texttt{\small jugal@illinois.edu}} &
    \href{mailto:aniket1602@gmail.com}{\texttt{\small aniket1602@gmail.com}}
  \end{tabular}
}
\date{}

\begin{document}
\pagenumbering{roman}
\maketitle
\thispagestyle{empty}
\begin{abstract}
Organizations increasingly seek to share and access datasets to improve their ML models and derive insights. Despite the immense demand for quality data, data exchange and collaboration have not reached their full potential, as cited by the European Commission~\cite{EC2020}. One of the key reasons is the lack of \textit{reciprocity}, where some participants perceive their contribution to others to be of higher value than what they receive in return \cite{reciprocityofDSinfra,barriersDS,factorsDS}.

We propose a general framework for data exchange without payments, focusing on two main properties: (i) \textit{reciprocal fairness}, which ensures that each agent receives utility proportional to their contribution to the total welfare -- contributions are quantifiable using well-established credit-sharing functions such as the Shapley share and proportional share, and (ii) \textit{core-stability}, which guarantees that no coalition of agents can identify an exchange among themselves which they all unanimously prefer to the current exchange. 

Surprisingly, we show that fair and stable exchanges exist for \emph{all} monotone continuous utility functions. Building on this, we investigate the computational complexity of finding approximate fair and stable exchanges. We present a local search algorithm for instances with monotone submodular utility functions and cross-monotone credit-sharing functions, e.g., Shapley share. We prove that this problem belongs to CLS ($=$ PPAD $\cap$ PLS) under mild assumptions. Our framework opens up several intriguing theoretical avenues for future research on data economics.
\end{abstract}
\clearpage
\pagenumbering{arabic}

\section{Introduction}\label{sec:intro} 
Data has proven to be an invaluable asset in the modern world due to its crucial role in improving ML models, decision making, and driving innovation. 
Apart from its use in competitive sectors such as retail~\cite{qi2020data}, analytics~\cite{troisi2018big}, supply chain~\cite{sanders2014big}, and banking~\cite{adeniran2024integrating}, data holds tremendous potential to advance the public good, particularly in areas like education \cite{HEDS}, scientific research \cite{nsf}, and public health \cite{ivankovic2023data}. 

However, such advancements crucially require sourcing a large and diverse pool of datasets from various organizations. For instance, improving ML models for healthcare require datasets that capture complex relationships between diseases, genetics, and socio-economic factors, for which the data collected by a single hospital cannot suffice~\cite{rieke2020future}. Likewise, the progress of autonomous vehicle technology depends on the sharing of data  across organizations, spanning different geographies, climates, and regulatory environments~\cite{xu2023federated}. However, data sharing is often fraught with concerns surrounding unfair contributions or the loss of privacy. As a result, as noted by the European Commission \cite{EC2020}, data exchange and collaboration has unfortunately not reached its full potential despite its immense benefits.

In an article on improving access and sharing of data \cite{oecd}, the Organization for Economic Cooperation and Development states that:

\begin{quote}
``\emph{Data sharing assumes common interests between the entities agreeing to share their data, including the interest and expectation that data holders can become data users and vice versa. Data sharing therefore can come with an expectation of some kind of \textbf{reciprocity} among the stakeholders engaged in data-sharing agreements.}"
\end{quote}

This highlights that a major obstacle to organizations exchanging data is the lack of \textit{reciprocity} -- the expectation that each participant receives benefits proportional to their contributions. Reciprocity is a natural and highly desirable fairness principle for data exchange, and its absence
has been identified as a ``core barrier'' towards data sharing in many contexts, including public health \cite{barriersDS} and scientific research \cite{factorsDS}. For example, hospitals may hold back data from clinical trials due to concerns that they will not receive the same level of access to data in return, even when collaboration could lead to better outcomes for all involved parties \cite{sharingcaring}. 

The ``norm of reciprocity'' \cite{reciprocityofDSinfra} is particularly important in settings where traditional economic incentives are absent, e.g., organizations like hospitals or universities that are restricted from buying or selling data for profit, and yet would like to exchange data to improve their private ML models. Even beyond data exchange, studies from behavioral economics \cite{fehr2006economics,GINTIS2000311} suggest that reciprocity increases engagement and trust, thereby leading to increased voluntary participation and better outcomes. All of this underscores the importance of ensuring reciprocity in data exchange.

Since no data sharing represents a trivially reciprocal exchange, an additional criterion -- such as economic efficiency -- becomes necessary. For this, we turn to \textit{core-stability} -- the gold standard desideratum that enforces stability, efficiency, individual rationality, and robustness in co-operative economies. In a core-stable exchange, no coalition of agents has an incentive to deviate from the current exchange by identifying an alternative exchange among themselves that they all strictly prefer over the current exchange. 
With our two desiderata in mind, in this paper we (1) introduce a formal model of a data exchange economy where agents exchange their data among each other for mutual benefit without payments, and (2) investigate the existence and computation of reciprocally \textit{fair} and core-\textit{stable} exchanges.

\subsection{Model Description}
\paragraph{The Exchange Economy.} A data exchange economy $\mathcal D$ comprises of set $\N$ of $n$ agents. Agent $i$ owns dataset $D_i$, and is willing to exchange parts of $D_i$ with other agents. An exchange is represented by $\x \in [0,1]^{n \times n} $, where $x_{ij}$ represents the fraction\footnote{The notion of fractional data exchange, representing the fraction of data points from a dataset, is well-established in data economics~\cite{ChenLX22} and incentives in learning~\cite{KGJ22, BlumHPS21}. Notably, a fractional exchange can also be interpreted as a distribution over integral exchanges (datasets assigned as a whole), allowing us to formulate the guarantees of our paper using expected utility and expected contributions.} of $D_i$ shared with agent $j$. Agent $i$ has a utility function $u_i \colon [0,1]^n \rightarrow \mathbb{R}_{\geq 0}$ which captures the value an agent has for the data bundle $\x_i = (x_{1i}, x_{2i}, \dots , x_{ni})$ she gets from an exchange. We assume that the utility functions are \emph{monotone} and \emph{continuous}. 
\smallskip

\noindent The data exchange is operated through a trusted central server, similar to data market economies \cite{AgarwalDS19} and decentralized machine learning paradigms like federated learning~\cite{mcmahan2017communication}, with which all agents share their prediction tasks. If the application demands it, the exchange can be conducted without explicit data sharing: agent $i$ can share the gradient of agent $j$'s ML model on an $x_{ij}$ fraction of agent $i$'s dataset. This is similar to federated learning methods and can alleviate privacy concerns of sharing data. Moreover, this disallows re-distribution of data shared between agents. The server can train ML models on data samples collected from different agents. It can therefore estimate an agent's prediction gain/utility on any data bundle, i.e., can compute $u_i(\x_i)$ for any $\x_i \in [0,1]^n$ and any agent $i$ efficiently.

\paragraph{Desiderata for the Exchange.} We formally introduce the concepts of reciprocal fairness and core-stability.
\begin{itemize}[leftmargin=*]
\item \emph{Reciprocal Fairness:} In a reciprocal exchange $\x$, the utility of an agent should be at least the sum of her contributions to other agents utilities, i.e., for all $i\in\N$, $u_i(\x_i) \geq \sum_{j} \psi_{ij}(\x_j)$ , where $\psi_{ij}(\x_j)$ measures agent $i$'s contribution to $u_j(\x_j)$. The contribution $\psi_{ij}(\x_j)$ is estimated through standard credit-sharing rules. In this paper, we only require the following three minimal assumptions about $\psi_{ij}(\cdot)$:
\begin{itemize}
    \item \emph{Monotonicity}, implying that $\psi_{ij}(\x_j)$ is weakly monotone in $x_{ij}$,
    \item \emph{Normalization}, implying that $\psi_{ij}(\x_j) =0$ if $x_{ij}=0$, and
    \item \emph{Efficiency}, implying that $\sum_i \psi_{ij}(\x_j) = u_j(\x_j)$ for all $j$.  
\end{itemize}

Our assumptions are satisfied by popular credit-sharing rules, such as the Shapley share~\cite{shapley1951notes, shapley1953value}.

\item \emph{Core-Stability:} An exchange $\x$ is core-stable if there exists no subset of agents $S \subseteq \N$ and no exchange $\y$ among agents in $S$ such that $u_i(\y_i) > u_i(\x_i)$ for all $i \in S$. We highlight a subtle but important feature of our core-stability desiderata: we do not impose any fairness requirement on the deviating coalition, while we impose the same on our solution, i.e., for a coalition to deviate from the exchange $\x$, it suffices for them to identify an exchange $\y$ (in particular, the exchange where all agents in the coalition share everything within themselves) which gives all agents in the coalition strictly higher utility than their utility in the exchange $\x$. \emph{In particular, the deviating coalition is allowed to conduct unfair exchanges, while our desired exchange is not!} 
\end{itemize}

\paragraph{Main Question.} It is trivial to determine an exchange satisfying only one of the foregoing desiderata: an exchange where all agents share nothing is trivially fair, and an exchange where all agents share everything is trivially core-stable, as there exists no exchange where any agent can get strictly higher utility. However, determining an exchange satisfying both the desiderata is non-trivial. 

To build some intuition on the non-triviality of the problem, we highlight a problem that arises when we approach the problem using a natural, iterative algorithm: (1) Maintain a fair exchange $\x$, (2) if $\x$ is not core-stable, identify a deviating coalition $S$ and a deviating exchange $\y$, (3) use $S$ and $\y$ to compute another fair exchange $\x'$, where we strictly improve the utilities of the agents in $S$, without reducing the utilities of the agents in $\N \setminus S$, and (4) update $\x \gets \x'$, and repeat until we get a fair, core-stable exchange. The main problem is that even if we were able to find an exchange that satisfies the fairness condition for all agents in $S$ while still improving their utility, it may be impossible to maintain the fairness condition for agents outside of $S$. This is primarily due to the  \emph{non-separability of $\psi_{ij}(\cdot)$}, i.e., an agent $i$'s contribution $\psi_{ij}(\x_j)$ in $j$'s utility is not only dependent on $x_{ij}$, but on all $x_{i' j}$ where $i' \neq i$. For instance, if datasets $D_i$ and $D_{i'}$ are similar, then increasing $x_{i'j}$ can decrease $\psi_{ij}(\x_j)$ and if $D_i$ and $D_j$ are complimentary, increasing $x_{i'j}$ can increase $\psi_{ij}(\x_j)$. This implies that any local exchanges done within agents in $S$ can potentially violate the fairness condition for agents in $\N \setminus S$. See Figure~\ref{Fig:separable} for an illustration.

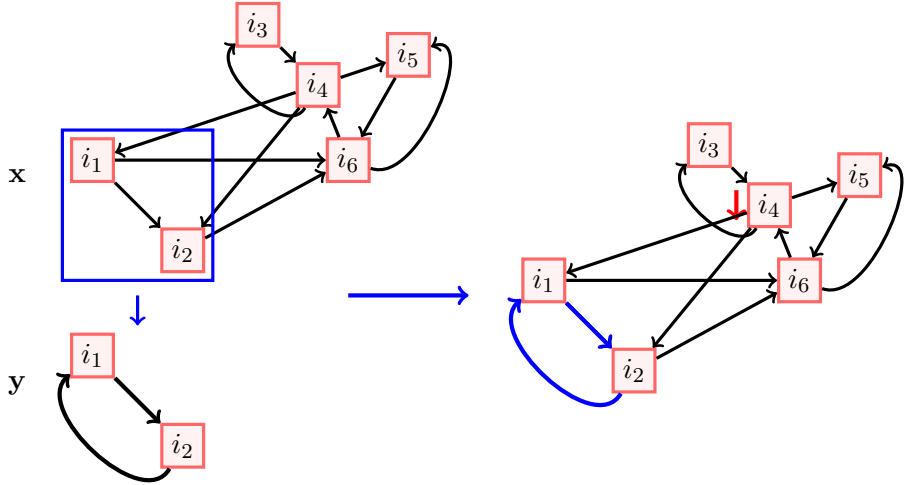
\begin{figure}
\centering 
\begin{tikzpicture}[scale=0.7,
roundnode/.style={circle, draw=green!60, fill=green!5, very thick, minimum size=7mm},
squarednode/.style={rectangle, draw=red!60, fill=red!5, very thick, minimum size=5mm},
]

{

\node at (-2,0.5) {$\mathbf{x}$};

\node [squarednode] at (-0.75,0.75) (a1) {$i_1$};
\node [squarednode] at (0.75,-0.75) (a2) {$i_2$};
\node [squarednode] at (2,3) (a3) {$i_3$};
\node [squarednode] at (3,2) (a4) {$i_4$};
\node [squarednode] at (4.5,2.5) (a5) {$i_5$};
\node [squarednode] at (3.5,0.75) (a6) {$i_6$};

\draw[->, black, very thick] (a3)--(a4);
\draw[->, black, very thick] (a4)--(a5);
\draw[->, black, very thick] (a5)--(a6);
\draw[->, black, very thick] (a6)--(a4);
\draw[->, black, very thick] (a4)--(a1);
\draw[->, black, very thick] (a1)--(a6);
\draw[->, black, very thick] (a2)--(a6);
\draw[->, black, very thick] (a4)--(a2);
\draw[->, black, very thick] (a1)--(a2);
\draw[->, black, very thick] (a4) to[out=-120,in=220] (a3);
\draw[->, black, very thick] (a6) to[out=-20,in=20] (a5);}

{
\node at (-2,-3) {$\mathbf{y}$};
\draw[blue,  very thick] (-1.25,1.25) rectangle (1.25,-1.25);
\draw[blue, very thick,->] (0,-1.5)--(0,-2);
\node [squarednode] at (-0.75,0.75-3.25) (b1) {$i_1$};
\node [squarednode] at (0.75,-0.75-3.25) (b2) {$i_2$};
\draw[->, black, ultra thick] (b1)--(b2);
\draw[->, black, ultra thick] (b2) to[out=-120,in=220] (b1);
}

{

\draw[blue,->, ultra thick] (3.5,0.5 -2)--(5.5,0.5-2);

\node [squarednode] at (-0.75+ 7.5,0.75-2) (c1) {$i_1$};
\node [squarednode] at (0.75+ 7.5,-0.75-2) (c2) {$i_2$};
\node [squarednode] at (2+ 7.5,3-2) (c3) {$i_3$};
\node [squarednode] at (3+ 7.5,2-2) (c4) {$i_4$};
\node [squarednode] at (4.5+ 7.5,2.5-2) (c5) {$i_5$};
\node [squarednode] at (3.5+ 7.5,0.75-2) (c6) {$i_6$};

\draw[red,->, ultra thick] (2.45+7.5,2.25-2)--(2.45+7.5,1.75-2);

\draw[->, black, very thick] (c3)--(c4);
\draw[->, black, very thick] (c4)--(c5);
\draw[->, black, very thick] (c5)--(c6);
\draw[->, black, very thick] (c6)--(c4);
\draw[->, black, very thick] (c4)--(c1);
\draw[->, black, very thick] (c1)--(c6);
\draw[->, black, very thick] (c2)--(c6);
\draw[->, black, very thick] (c4)--(c2);
\draw[->, black, very thick] (c1)--(c2);
\draw[->, black, very thick] (c4) to[out=-120,in=220] (c3);
\draw[->, black, very thick] (c6) to[out=-20,in=20] (c5);

\draw[->, blue, ultra thick] (c1)--(c2);
\draw[->, blue, ultra thick] (c2) to[out=-120,in=220] (c1);

}

\end{tikzpicture}
\caption{Illustration of the issue of non-separability. The squared nodes correspond to the agents and an edge ${(i,j)}$ between agents $i$ and $j$ imply that $x_{ij} > 0$, i.e., $i$ shares some of her data with $j$. Given a reciprocal exchange $\x$, say we have a coalition $S$ ($S = \{i_1,  i_2\}$ in our example) such that all agents in $S$ prefer $\y$ to $\x$. Performing any local update to the exchange in $\x$ can affect the utility-contributions ($\psi_{ij}(\cdot)$s) of all agents that send data flow to agents in $S$, e.g., increasing data-flow from $i_1$ to $i_2$ decreases $i_4$'s contribution.}
\label{Fig:separable}
\end{figure}

\paragraph{Comparison with Classic Exchange Economies.} Exchange markets, also known as Arrow-Debreu markets~\cite{walras2013elements}, have been extensively studied in economics and theoretical computer science. In these settings, agents with endowments of goods engage in trades to maximize their utility. The final allocation is obtained through a \emph{competitive equilibrium} that determines prices for each good, and an exchange, such that demand meets supply for each good. In contrast, data exchange economy differs fundamentally due to the fact that data is a \emph{replicable asset}: it can be duplicated at zero marginal cost and therefore has no concept of limited supply. As a result, applying competitive equilibrium as a solution concept becomes less suitable for data exchange economies. 

\subsection{Our Contributions}
In classic exchange markets, fairness\footnote{Fairness is guaranteed through equilibrium prices. In particular, each agent gets a utility-maximizing bundle subject to the total price of her initial endowment, i.e., the higher the price of the initial endowment, the higher is the agent's utility.} and core-stability are consequences of the competitive exchange (CE). The existence of a CE is only guaranteed under quasi-concave utility functions~\cite{arrow1954existence}. It is only natural to seek existence under similar guarantees. However, as our first main result, we show the existence of fair and core-stable exchange much more generally. 

\begin{theorem}
\label{thm-intro:informal}
 A reciprocal and core-stable exchange exists for all monotone continuous utility functions and any credit-sharing functions satisfying monotonicity, normalization, and efficiency. 
\end{theorem}

Theorem~\ref{thm-intro:informal} is intriguing, given that the set of reciprocal exchanges, and the set of core-stable exchanges are both non-convex and therefore not well-behaved for the application of standard fixed point theorems. Furthermore, applying fixed point theorems such as the \emph{Scarf Lemma}~\cite{scarf1967core} ---  used to prove the existence of core stability in matching~\cite{gale1962college}, and housing allocation~\cite{shapley1974cores} --- would require a weaker definition of core stability. In particular, the deviating coalition must also comply with the reciprocity constraint, as Scarf's lemma ensures the absence of such coalitions only within a defined domain, which here refers to a subset of reciprocal exchanges. Consequently, whatever constraints apply to exchange in consideration, also apply to the deviating exchanges. In contrast, our techniques offer broad applicability, particularly in scenarios where core stability is required alongside other constraints, and deviating coalitions are not required to meet these constraints. We prove \Cref{thm-intro:informal} in \Cref{sec:existence}.

In this paper, we also initiate the study on the computational complexity of $\varepsilon$-core-stable and $\varepsilon$-fair exchanges (obtained by adding an additive error of $\varepsilon$ to reciprocity and core-stability).  Our existence proof (Section~\ref{sec:existence}) invokes Brouwer's fixed point theorem on a function $f$ to show the existence of a core stable and fair exchange. However, the domain in our fixed point formulation is the intersection of exponentially many hyperplanes and therefore does not admit an explicit polynomial representation. We thus embed the domain in our fixed point formulation to a \emph{box}, and define a new fixed point function $\tilde{f}$, such that $\varepsilon$-approximate fixed points of $\tilde{f}$ can be mapped to  $\textup{poly}(\varepsilon)$-fair and $\textup{poly}(\varepsilon)$-core stable exchange in polynomial time. This puts the problem in PPAD. We suspect that for better computational results, one needs to make more assumptions on the utility and the share functions. As an initiation, we consider monotone $L$-Lipschitz utility functions and \emph{cross-monotone share functions}, i.e., $\psi_{ij}(\x_j)$ is monotone non-increasing in $x_{i'j}$ for all $i' \neq i$~\cite{BhaskaraGIKMS24}. This condition is satisfied when we consider $L$-Lipschitz continuous submodular utility functions, and Shapley shares to measure $\psi_{ij}(\cdot)$. For instances satisfying the foregoing conditions, we give a \emph{local search algorithm} (discussed in \cref{sec:comp-main} and \cref{sec:comp}), that determines a $\varepsilon$-core stable and $\varepsilon$-fair exchange. When the utility functions are $L$-Lipschitz, and  $L/ {\varepsilon} = \textup{poly}(n)$, then our local search algorithm establishes that the problem is in PLS, which shows membership in the class CLS = PPAD $\cap$ PLS ~\cite{FearnleyGHS23}. Our key technical contribution lies in proving the convergence of local search algorithm, even when the utility and credit-sharing functions are non-separable (changing a single $x_{ij}$ can affect the reciprocity condition globally). 

\begin{theorem}
For monotone $L$-Lipschitz utility functions, and cross-monotone share functions, there exists a local-search procedure that determines an $\varepsilon$-core stable and $\varepsilon$-fair exchange. For $L/ {\varepsilon} = \textup{poly}(n)$, finding an $\varepsilon$-core stable and $\varepsilon$-fair exchange is in CLS.
\end{theorem}

\noindent {The above theorem contrasts computational results on finding solutions in classic exchange economies with utility functions that mildly generalize linear utilities. It is well known that computing a constant-factor approximation of an equilibrium in Arrow-Debreu markets with additively separable piecewise linear concave (SPLC) utilities is already PPAD-hard, where $L/\varepsilon$ is constant~\cite{chen2009settling, Rubinstein18}. Moreover, this problem remains PPAD-hard even for the special case of Fisher markets~\cite{ChenT09,VaziraniY11,DeligkasFHM24}.}  We offer a high-level overview of our approach in \Cref{sec:comp-main} and provide a more detailed explanation in \Cref{sec:comp}.

\subsection{Related Work}
Our work borrows concepts, techniques, and formulation broadly from data economics, classical microeconomic theory on markets, and modern decentralized collaborative ML paradigms. Each of the preceding fields is a vibrant area of research in its own right, making it impossible for us to provide a comprehensive survey. Instead, we highlight only parts of the literature that closely relate to our study.

\paragraph{Data Markets.} Data markets are two-sided real-time platforms facilitating pricing and selling data to data-seekers. 
There is a long line of work~\cite{admati1986monopolistic, admati1990direct,  bergemann2018design, BabaioffKP12} that investigates revenue-maximizing strategies of a monopolist data seller. 
Another line of work explores equilibrium and auctions in datamarkets in the presence of externalities~\cite{AgarwalDHR20, Hossain024}.

\paragraph{Classic Exchange Markets.} Exchange economies was introduced by Leon Walras~\cite{walras2013elements} in 1874. The canonical exchange happens through the notion of a \emph{competitive equilibrium (CE)}, which sets prices for each good such that its demand equals supply.  The existence of a CE in exchange economies was proven by Arrow and Debreu in 1954~\cite{arrow1954existence}. Since then, there has been a long line of work on the properties structure of CE (e.g., convex programming formulations~\cite{devanur2016rational}) and its computation complexity-- In particular, polynomial time algorithms are known for linear utilities~\cite{DuanM15, ChaudhuryM18, DuanGM16, GargV23}, and WGS utilities ~\cite{codenotti2005market, garg2019auction}, and beyond WGS is essentially PPAD-hard~\cite{chen2017complexity, chen2009settling, codenotti2006leontief, garg2017settling}. 

\paragraph{Federated Learning.} 
 Federated learning (FL) provides a privacy-preserving effective distributed learning paradigm where a group of agents holding local data samples can train a joint machine learning model~\cite{mcmahan2017communication}. 
 \emph{Data exchange} can be seen as a \emph{private learning paradigm}, where each agent exchanges their own data for other valuable data to train their own private ML model, while the latter (FL) is a \emph{public learning paradigm} where all agents train the same ML model from data shared by the federating agents.  ~\cite{donahue2021optimality} work on coalition formation in FL, but they assume that all agents have the same utility function (learning objective), but have biases in their local distributions which create sub-optimal local models. The primary distinction to our setting is that we assume agents have heterogeneous learning objectives. Broadly speaking, FL has witnessed substantial integration of concepts from game theory. In particular, Stackelberg games~\cite{khan2020federated, pandey2019incentivize}, non co-operative games~\cite{zou2019dynamic, cheng2021dynamic}, auctions~\cite{roy2021distributed}, budget balanced reward mechanisms~\cite{MurhekarYCLM23} have been adopted for incentivizing participation in FL.

\paragraph{Independent Work.} Independent of our work,~\cite{BhaskaraGIKMS24} also presents an intriguing study on data exchange, addressing computational aspects of welfare maximization under reciprocity constraints and the existence of reciprocal and core-stable exchanges. Their latter result shares the same spirit as ours, though with key differences in model formulation. Their model can be viewed as a special case of ours because $(i)$ their definition of core-stability allows only \emph{reciprocal exchanges} within deviating collations -- meaning if $\x$ is core-stable, no coalition $S$ and no \emph{reciprocal exchange} $\y$ among agents in $S$ exists, such that every agent in $S$ prefers $\y$ over $\x$. In contrast, our model does not impose such a constraint on deviating exchanges $\y$. $(ii)$ Moreover, their exchange mechanism is randomized, where each agent's data allocation represents a distribution over datasets owned by other agents. Consequently, their utility and  share functions are interim -- reflecting expected gains and contributions from a randomized exchange -- and are more specialized than the general utilities and share functions we consider.\footnote{In fact, their functions are linear in terms of $x_{iS}$, where $x_{iS}$ is the probability that agent $i$ receives datasets from agents in $S$. We remark that one can model $x_{iS} = \prod_{j \in S}x_{ij} \cdot \prod_{j \notin S} (1-x_{ij})$, and have the desired utility function in our model as a function of $\x$.}

Furthermore, they establish the existence of reciprocal and $\varepsilon$-core stable exchanges, while we show the existence of exact reciprocal and exact core-stable exchanges in more general settings. Regarding computational results, their work is primarily centered on welfare maximization, whereas our focus is on identifying reciprocal and core-stable exchanges.
\section{Preliminaries}\label{sec:prelim}
We begin by describing our model of data exchange. We assume there is a set $\N$ of $n$ agents, and each agent $i$ brings a dataset $D_i$ to a centralized data exchange platform. We allow the datasets to be fractionally exchanged between all possible pairs of agents, and hence represent an \textit{exchange} as a matrix $\x \in [0, 1]^{n\times n}$, with $x_{ij}$ denoting the fraction of dataset $D_i$ that agent $i$ gives to agent $j$. We let $\x_i = (x_{1i}, x_{2i}, \dots, x_{ni})$ denote the \textit{bundle} of datasets made available to agent $i$ in the exchange $\x$. An agent $i$ has an associated utility function $u_i : [0,1]^n \rightarrow [0,1]$ which captures the utility $u_i(\x_i)$ agent $i$ receives from the bundle $\x_i$ in the exchange $\x$; for simplicity we interchangeably use $u_i(\x)$ and $u_i(\x_i)$. We nominally assume that utility functions are continuous, \textit{normalized}, i.e., $u_i(0^{n\times n}) = 0$, and \textit{monotone}, i.e., $u_i(\cdot)$ is non-decreasing in $x_{ji}$ for all $i,j \in \N$.

\begin{definition}[Data Exchange Problem] \normalfont The input to the data exchange problem is a tuple $\I = \langle \N, \D, \Val \rangle$, where $\N=[n]$ is a set of agents, $\D = (D_1, \ldots, D_n)$ is a list of agents' datasets, with $D_i$ denoting the dataset of agent $i$, and $\Val = (\val_1, \ldots, \val_n)$ is a list of agent utility functions, where $\val_i : [0,1]^{n\times n} \rightarrow [0, 1]$. The output is an exchange $\alloc \in [0,1]^{n\times n}$ where $x_{ij}$ denotes the fraction of dataset $D_i$ that agent $i$ gives to agent $j$.
\end{definition}

We define the desiderata for our exchanges, beginning with our notion of fairness.

\subsection{Utility sharing functions and reciprocal fairness} 
Intuitively, we regard an exchange to be fair, or \textit{reciprocal}, if every agent receives as much utility from the exchange as their \textit{contribution} to other agents in the exchange. Formalizing this intuition necessitates defining the contribution of an agent towards the utility of another agent in the exchange, for which we define utility sharing functions.  

\paragraph{Share functions.}
Conceptually, the utility $u_j(\x)$ of an agent $j$ in an exchange $\x$ can be viewed as a sum of contributions of the other agents given by $\sum_{i\in \N} \psi_{ij}(\x_j)$. Here, $\psi_{ij} : [0, 1]^n \rightarrow [0, 1]$ is a \textit{sharing} rule that measures the contribution of agent $i$ towards the utility of agent $j$ given that $j$ receives the bundle $\x_j$; we interchangeably use $\psi_{ij}(\x)$ and $\psi_{ij}(\x_j)$. 

\begin{definition}[Utility sharing function]\label{def:share-function} \normalfont For an agent $i$, a collection of functions $\{\psi_{ij}\}_{j\in \N}$ where $\psi_{ij}:[0,1]^n \rightarrow [0,1]$ are considered share functions if for any exchange $\alloc$ and agent $i\in \N$, we have 
\begin{enumerate}
\item[(i)] (Monotonicity) $\psi_{ij}(\x_j)$ is non-decreasing in  $x_{ij}$ for every $j\in \N$,
\item[(ii)] (Normalization) $\psi_{ij}(\x_j) = 0$ if $x_{ij} = 0$ for every $j\in \N$,
\item[(iii)] (Efficiency) $u_j(\x_j) = \sum_{i\in \N}\psi_{ij}(\x_j)$.
\end{enumerate}
\end{definition}

We discuss some examples of sharing functions that have been studied in the literature. 

\paragraph{Shapley share.} For a subset $S\subseteq \N$, let $\x_j[S] \in [0,1]^n$ denote the restriction of the bundle $\x_j$ to the set $S$, i.e., $(\x_j[S])_k = x_{kj}$ if $k\in S$ and $0$ otherwise. The Shapley share measures the contribution of an agent $i$ to agent $j$ in an exchange $\x$ as follows. For every possible subset $S\subseteq \N \setminus \{i\}$ of agents, we compute the marginal gain in the utility of $j$ obtained by adding $i$ to the subset $S$, i.e., $\mu(S) = u_j(\x_i[S\cup \{i\}]) - u_j(\x_j[S])$. The Shapley share $\psi_{ij}$ is then the average marginal gain $\mu(S)$ taken over all possible such subsets $S$. Formally,
\begin{equation}\label{eq:shapley-share}
\psi_{ij}(\alloc_j) = \frac{1}{n}\sum_{S\subseteq \N\setminus\{i\}} \binom{n-1}{|S|}^{-1} \cdot (u_j(\alloc_j[S\cup \{i\}]) - u_j(\alloc_j[S])).
\end{equation}
We prove in Appendix~\ref{app:prelim} that the Shapley share functions as defined above are share functions as per Definition~\ref{def:share-function}.
\begin{restatable}{proposition}{propshap}\label{prop:shap}
For every $j\in \N$ and exchange $\x$, we have $u_j(\x_j) = \sum_{i\in \N} \psi_{ij}(\x_j)$.
\end{restatable}

Another commonly used share function is the \emph{proportional share}, which attributes agent $i$'s contribution to agent $j$ in an exchange $\x$ as a fixed fraction of $j$'s utility, $u_j(\x_j)$, proportional to a weight $w_{ij}$~\cite{BhaskaraGIKMS24}.

\paragraph{Reciprocity.} Having defined utility sharing functions, we formally define \textit{reciprocal} exchanges. Given an exchange $\x$, we let $\Psi_i(\x)=\sum_{j \in \N}\psi_{ij}(\x_j)$ denote the total contribution of agent $i$ towards other agents in $\x$. Let $\Delta_i(\x) := \Psi_i(\x) - u_i(\x)$ denote the \textit{surplus} of agent $i$ in $\x$. By the definition of a sharing rules, we have that $\sum_{i\in \N} \Delta_i(\x) = 0$ for any exchange $\x\in[0,1]^{n\times n}$.
\begin{definition}[Reciprocal exchange]\label{def:reciprocity} An exchange $\alloc$ is said to be $\delta$-reciprocal if for all agents $i\in \N$, $|\Delta_i(\alloc)| \le \delta$. A $0$-reciprocal exchange is called a reciprocal exchange.
\end{definition}

Note that the exchange $\alloc = 0^{n\times n}$ is trivially reciprocal. 

\subsection{Core-stability}
Next, we define another desideratum of exchange called \textit{core-stability}, which captures notions of \textit{stability} and \textit{efficiency}. Intuitively, an exchange is core-stable if no subset of agents can form a coalition and exchange data only among themselves such that every agent in the coalition gets better utility. 
\begin{definition}\label{def:core-stable}
An exchange $\alloc$ is $\eps$-\cs if and only if for all $\N' \subseteq \N$, there is no exchange $\x'$ of the instance $\langle \N', (D_i)_{i\in \N'}, (\val_i)_{i\in \N'} \rangle$ such that $\val_i(\alloc'_i) > \val_i(\alloc_i) + \eps$ for all $i \in \N'$. A $0$-\cs exchange is called a \cs exchange.
\end{definition}

In an attempt to characterize core-stable exchanges, we consider the set of core-stable exchanges that can be generated in the same spirit as the core-stable housing allocations generated by the top trading cycle algorithm~\cite{shapley1974cores}. For this, we define the \textit{exchange graph} $G(\x) = (\N,E)$ of an exchange $\x$ as a directed graph with nodes as the set of agents, and $(i,j) \in E$ if and only if $x_{ij} < 1$, i.e., if agent $i$ does not provide its data entirely to agent $j$. We argue that that if $G(\x)$ is acyclic, then $\x$ is core-stable. Let $S$ be the set of sources in $G(\x)$. Agents in $S$ get all data from all other agents, implying that there exists no exchange that give them strictly better utility. Therefore, agents in $S$ will not be part of any deviating coalition. Now consider the set of agents in the topological order in $G(\x)$. Observe that no agent can be part of a deviating coalition unless their predecessors (agents higher up in the topological ordering) are part of the deviating coalition. This shows that no agent is a part of a deviating coalition (see Figure~\ref{fig:dependency}), and hence that $\x$ is \cs. We now formalize the above intuition, while extending the definitions and arguments to approximate-core-stability.

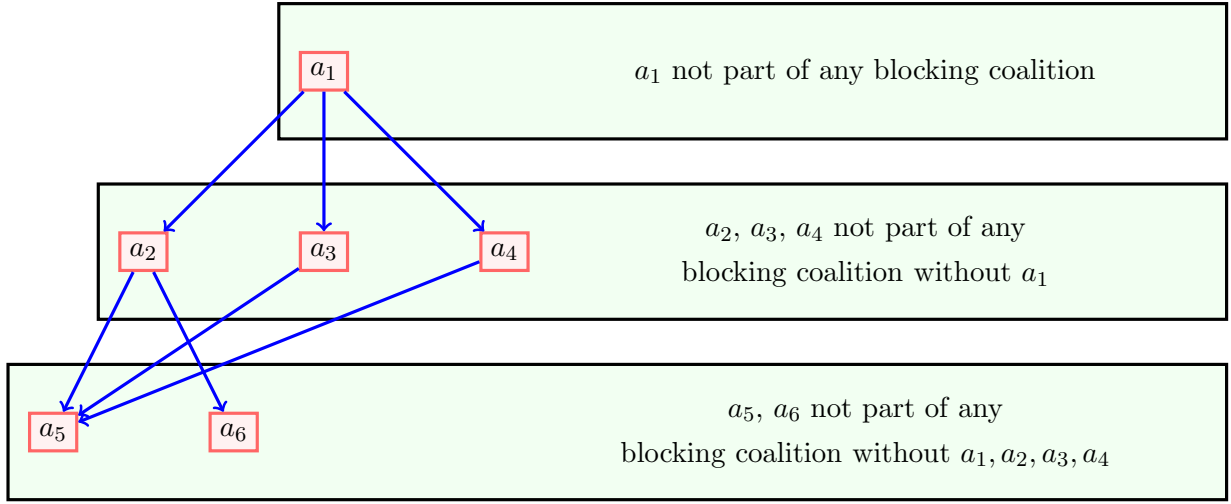
\begin{figure}
\centering 
\begin{tikzpicture}[scale=0.8,
roundnode/.style={circle, draw=green!60, fill=green!5, very thick, minimum size=7mm},
squarednode/.style={rectangle, draw=red!60, fill=red!5, very thick, minimum size=5mm},
]

\draw[black, fill=green!5, very thick] (-0.5,0.75) rectangle (10,-0.75);
\draw[black, fill=green!5, very thick] (-2.5,-1.25) rectangle (10,-2.75);
\draw[black, fill=green!5, very thick] (-3.5,-3.25) rectangle (10,-4.75);

\node[squarednode] at (0,0) (a1) {$a_1$};
\node[squarednode] at (-2,-2) (a2) {$a_2$};
\node[squarednode] at (0,-2) (a3) {$a_3$};
\node[squarednode] at (2,-2) (a4) {$a_4$};
\node[squarednode] at (-3,-4) (a5) {$a_5$};
\node[squarednode] at (-1,-4) (a6) {$a_6$};

\draw[blue,->, very thick] (a1)--(a2);
\draw[blue,->, very thick] (a1)--(a3);
\draw[blue,->, very thick] (a1)--(a4);
\draw[blue,->, very thick] (a2)--(a5);
\draw[blue,->, very thick] (a2)--(a6);
\draw[blue,->, very thick] (a3)--(a5);
\draw[blue,->, very thick] (a4)--(a5);

\node at (6,0) {$a_1$ not part of any blocking coalition};
\node at (6,-1.75) {$a_2$, $a_3$, $a_4$ not part of any};
\node at (6,-2.25)  {blocking coalition without $a_1$};
\node at (6,-3.75)  {$a_5$, $a_6$ not part of any};
\node at (6,-4.25)  {blocking coalition without $a_1, a_2, a_3, a_4$};

\end{tikzpicture}
\caption{Illustration of how an acyclic exchange graph will ensure core-stability.}\label{fig:dependency}
\end{figure}

\begin{definition}\label{def:exchange-graph}
Given an exchange $\alloc$ for an instance $(\N, \D, \mathcal{U})$ and a parameter $\alpha \in (0, 1)$, we define the \textit{exchange graph} $G(\alloc, \alpha)$ as a directed graph with the vertex set $V(G(\x, \alpha)) = \N$ and the edge set $E(G(\x, \alpha)) = \{(i, j) : x_{ij} < 1- \alpha\}$. 
\end{definition}

Intuitively, our goal is ensure that if an exchange $\x$ is not $\eps$-core-stable, then for some appropriately chosen parameter $\alpha(\eps)$ the exchange graph $G(\x, \alpha(\eps))$ contains a cycle $C$ which represents a coalition of agents who would rather share data among themselves. We choose the parameter $\alpha(\eps)$ carefully so that by increasing the `flow' along the edges of such a cycle $C$, we obtain an exchange $\y$ which each agent in $C$ prefers over $\x$ by an additive value of $\eps$. Thus, an exchange $\x$ whose exchange graph $G(\x,\alpha(\eps))$ has a cycle $C$ is not $\eps$-core stable, since agents in $C$ can significantly improve their utility (i.e. at least by $\eps)$) by sharing data among themselves. Conversely, if the exchange graph $G(\x, \alpha(\eps))$ is acyclic for an exchange $\x$, then $\x$ must be $\eps$-core-stable. We prove this formally.

Let $I_{ji} = \{\x \in [0, 1]^{n\times n} : x_{ji} = 1\}$ be the set of exchanges where $j$ shares her dataset entirely with $i$. Decreasing $x_{ji}$ reduces the utility of $i$. Given an exchange $\x \in I_{ji}$, we are interested in the maximum additive decrease in $x_{ji}$ which does not cause a significant drop in the utility of agent $i$. We define this quantity below. Let $\e_{ji} \in [0,1]^{n\times n}$ be the exchange where no data is shared among agents, except from $j$ to $i$. 
\begin{definition}\label{def:alpha-eps}
For $i, j \in \N$, we define $a_{ji}^\eps :I_{ji} \rightarrow [0,1]$ as the function given by 
\[
a_{ji}^\eps(\x) = \max\{ a\in [0,1] : u_i(\x - a\cdot \e_{ji}) \ge u_i(\x) - \eps/n\},
\]
that is, $a_{ji}^\eps(\x)$ is the maximum additive decrease $a \in [0,1]$ in $x_{ji}$ which decreases the utility of $i$ by additively by $\eps/n$. Define $\alpha(\eps) = \min_{j,i\in \N} \inf_{\x\in{I_{ji}}} a_{ji}^\eps(\x)$.
\end{definition}

By definition of $\alpha(\eps)$, we observe that:
\begin{observation}\label{obs:alpha}
For all $i,j \in [n]$ and $\alloc \in [0,1]^{n\times n}$ with $x_{ji} = 1$, we have 
\begin{align*}
\val_i(\x) \leq \val_i(\x - \alpha(\eps)\cdot \e_{ji}) + \eps/n.
\end{align*}
\end{observation}

\noindent We now prove that acyclic exchange graphs imply approximately-\cs exchanges.
\begin{lemma}\label{lem:acyclic-cs}
If $G(\x, \alpha(\eps))$ is acyclic for an exchange $\x$, then $\x$ is $\eps$-core-stable.
\end{lemma}
\begin{proof}
Suppose the claim does not hold and there is a subset $\N'$ of agents who can exchange data among themselves and improve their utility additively by $\eps$. Let $\x'\in[0,1]^{n\times n}$ be the exchange where agents in $\N'$ fully share data with each other, i.e., $x'_{ij} = 1$ for all $i,j\in \N'$ and $x'_{ij} = 0$ otherwise. Then, for every $i\in \N'$, we have $u_i(\x') > u_i(\x) + \eps$.

Since $G(\x, \alpha(\eps))$ is acyclic, there must be an agent $s\in \N'$ who is a source in the subgraph of $G(\x, \alpha(\eps))$ induced by $\N'$. Hence, $x_{is} \ge 1-\alpha(\eps)$ for all $i\in \N'\setminus \{s\}$. For simplicity, let us re-index the agents so that $\N' = \{1,2,\dots, |\N'|\}$. Then by repeatedly using \cref{obs:alpha}, we have:
\begin{align*}
\val_s(\x'_s) &\leq \val_s(\x'_s - \alpha(\eps)\cdot \e_{1s}) + \eps/n \tag{Observation \ref{obs:alpha}}\\
&\leq \val_s(\x'_s - \alpha(\eps)\cdot (\e_{1s} + \e_{2s})) + 2\cdot(\eps/n) \tag{Observation \ref{obs:alpha}}\\
&\ldots \\
&\leq \val_s(\x'_s - \alpha(\eps)\cdot \sum_{i\in \N'} \e_{is}) + |\N'|\cdot (\eps/n) \tag{Observation \ref{obs:alpha}}\\
&\leq \val_s(\x_s) + |\N'|\cdot (\eps/n) \tag{since $x_{is} \ge 1-\alpha(\eps)$ for all $i\in \N'\setminus\{s\}$}\\
&\leq \val_s(\x_s) + \eps, \tag{since $|\N'|\le n$}
\end{align*}
which contradicts the fact that $u_s(\x'_s) > u_s(\x_s) + \eps$. Thus, $\x$ must be $\eps$-core-stable.
\end{proof}
\section{Existence of \Fair~and Core-Stable Exchanges}\label{sec:existence}
In this section, we prove the existence of exchanges that are reciprocal and core-stable when agent utility functions and share functions are continuous and monotone. First, we show the existence of a reciprocal and $\eps$-\cs exchange for $\eps \in (0, 1)$.

\begin{theorem}\label{thm:fair-cs-existence}
For any $\eps \in (0,1)$, any data exchange instance $\I$ where utility functions and sharing functions are continuous and monotone admits a reciprocal and $\varepsilon$-\cs exchange.
\end{theorem}

We use fixed point theorems to prove \cref{thm:fair-cs-existence}, which state that a function $g: Z \rightarrow Z$ admits a fixed point $z$ (i.e. $g(z) = z$) under certain conditions on $g$ and $Z$. For instance, Brouwer's fixed point theorem -- which we use -- requires that $Z$ is convex and compact and $g$ is continuous. Most fixed point theorems require convexity and compactness of the domain, and continuity-like properties for the function/correspondence. Since our exchange $\x$ needs to satisfy two properties, namely fairness and $\eps$-core-stability, a natural attempt would be to use $Z$ to capture one property and $g$ to capture the other, e.g., we could define $Z$ to be the set of all fair exchanges and $g$ such that every fixed point of $g$ corresponds to a $\eps$-core-stable exchange. Unfortunately, the set of all fair exchanges and the set of all $\eps$-core-stable exchanges is non-convex!  Therefore, our strategy is to (i) identify a non-trivial subset of $\eps$-core-stable exchanges $Z$ which is convex or at least \emph{homeomorphic} to a convex set, and (ii) define a continuous function $g: Z \rightarrow Z$, such that the fixed points of $g$ correspond to a fair exchange. Then, Brouwer's fixed point theorem guarantees that $g$ admits a fixed point $z^*$, which we show corresponds to an exchange that is reciprocal and $\eps$-\cs.

\subsection{Defining the set \texorpdfstring{$Z$}{Z} to capture core-stable exchanges}
Unfortunately, the set of $\eps$-\cs exchanges is non-convex and does not have an algebraically tractable form. Therefore, we leverage \cref{lem:acyclic-cs} and identify its non-trivial subset $X_\eps$ comprising of exchanges $\x$ s.t. for any cycle of agents $C$ in the complete graph $K_n$, we have $\prod_{e \in C} (1-x_e) \le \alpha(\eps)^n$. Note that any $\x\in X_\eps$ is $\eps$-\cs, since the constraint ensures that for each cycle $C$ there must exist some edge $(i,j)\in C$ s.t. $1-x_{ij} \le \alpha(\eps)$, thereby implying that $G(\x,\alpha(\eps))$ is acyclic. However, $X_\eps$ still does not satisfy convexity which is required for our fixed point formulation.

Despite its non-convexity, $X_\eps$ has desirable properties: \emph{there exists a {homeomorphism} $h$ from $X_{\varepsilon}$ to a convex compact set $Z$}, described as follows. Let $M > cn \log(1/\alpha(\eps))$ for some large constant $c > 1$. We define $Z$ to be a subset of $[0, M]^{n\times n}$, and associate each point in $[0, M]^{n\times n}$ with an exchange in $[0, 1]^{n\times n}$. Concretely, we define the function $f: [0, M]^{n\times n} \rightarrow [0, 1]^{n\times n}$ as $f(z)_{ij} = 1 - \exp(-z_{ij})$, which maps the point $z$ to the exchange $f(z)$. Equivalently, the homeomorphism  $h \colon X_{\varepsilon} \rightarrow Z$ is $h = f^{-1}$, that is, $(h(\x))_{ij} = z_{ij} = \log(1/(1-x_{ij}))$ for all $i,j\in \N$. Using the definition of $X_\eps$, we then define the set $Z$ as follows: 
\begin{equation}\label{eq:set-z}
\begin{aligned}
Z = \{z \in [0, M]^{n\times n} : \text{for all } & \text{cyclic permutations } (i_1, \dots, i_k) \text{ of } \N, \sum_{j=1}^k z_{i_j i_{j+1}} \ge n \log (1/\alpha) \},
\end{aligned}
\end{equation}
where the definition used the convention that $i_{k+1} = i_1$. 

First note that by choice of $M > cn\log(1/\alpha)$, an element $z \in [0,M]^{n\times n}$, where $z_{ij} = M$ for all $i, j\in [n]$ satisfies $\sum_{j=1}^k z_{i_j i_{j+1}} \ge n\log(1/\alpha)$ for all cyclic permutations $(i_1, \dots, i_k)$ of $[n]$, and hence lies in $Z$. Thus, $Z$ is non-empty. Since $Z$ is defined by a set of linear constraints, and is contained in $[0, M]^{n\times n}$, $Z$ is a polytope. Thus, $Z$ is non-empty, convex, and compact and meets the domain requirements for Brouwer's fixed point theorem.

\begin{restatable}{observation}{lemZconvex}\label{obs:convexcompact}
$Z$ is a convex and compact set.
\end{restatable}

Next, we formally show that for every point $z \in Z$, $f(z)$ corresponds to a $\eps$-\cs exchange, implying that we have successfully identified a subset of $\eps$-core-stable exchanges $X_\eps$ that is homeomorphic to a convex compact set.

\begin{lemma}\label{lem:z-cs}
For all $z \in Z$, the exchange $f(z)$ is $\varepsilon$-\cs.
\end{lemma}
\begin{proof}
Consider an arbitrary $z \in Z$. We prove that $f(z)$ is $\eps$-\cs by showing that the exchange graph $G(f(z), \alpha)$ is acyclic and then invoking \cref{lem:acyclic-cs}. Consider any cyclic permutation $(i_1, \ldots, i_k)$ of $[n]$. By definition of $Z$, we have $\sum_{j \in [k]} z_{i_j i_{j+1}} \geq n\log (1 / \alpha)$. Since $k\le n$, this implies that for some $j \in [k]$, $z_{i_j i_{j+1}} \geq \log (1 / \alpha)$. By definition of $f$, we obtain that: 
\[
\log(\frac{1}{1-f(z)_{i_j i_{j+1}}}) \ge \log(\frac{1}{\alpha}),
\]
which implies $f(z)_{i_j i_{j+1}} \ge 1 - \alpha$. This means that the exchange graph $G(f(z), \alpha)$ does not have the edge from $i_j$ to $i_{j+1}$, and hence the cyclic permutation $(i_1, \dots, i_k)$ can correspond to a cycle in $G(f(z), \alpha)$. Thus, $G(f(z), \alpha)$ is acyclic, and \cref{lem:acyclic-cs} implies that $f(z)$ is $\varepsilon$-\cs.
\end{proof}

Therefore, we have a convex, compact set $Z\neq\emptyset$, every point of which can be mapped uniquely to an approximate core-stable allocation. This will be the domain for our fixed point formulation. 

\subsection{Defining the function \texorpdfstring{$g: Z\rightarrow Z$}{g: Z to Z} to capture reciprocity} Having defined $Z$, we proceed to define a continuous function $g: Z\rightarrow Z$ whose fixed points correspond to reciprocal exchanges. Thus, for a fixed point $z$ of $g$, we would like to have $\Delta_i(f(z)) = 0$ for all $i\in \N$. By the efficiency property of share functions, we have $\sum_{i} \Delta_i(f(z)) = 0$ for all $z \in Z$. Thus, it suffices to show that there exists a $z \in Z$, with $\Delta_i(f(z)) = \Delta_j(f(z))$ for all $i,j$.

To define $g$, consider a point $z\in Z$ s.t. $\Delta_i(f(z)) > \Delta_j(f(z))$ for agents $i,j\in\N$, i.e., agent $i$ contributes more utility into the exchange than $j$. Thus, it is intuitive that to equalize the surpluses of $i$ and $j$, our function $g$ should map $z$ to a point $z'$ by decreasing $z_{ij}$ and increasing $z_{ji}$ (since that corresponds to decreasing $f(z)_{ij}$ and increasing $f(z)_{ji}$). However, this increase/decrease must be done \textit{as long as $z'\in Z$}. For this, we define functions which capture that maximum increase/decrease in $z_{ij}$ which still maintains feasibility.

\paragraph{Traversing to the boundary of $Z$.} We define functions $\beta^+_{ij}$ and $\beta^-_{ij}$ as follows:

\begin{itemize}
\item $\beta^+_{ij}(z) = \arg\max\{b\ge 0: (z + b\cdot \e_{ij}) \in Z$\}
\item $\beta^-_{ij}(z) = \arg\max\{b\ge 0: (z - b\cdot \e_{ij}) \in Z$\}
\end{itemize}

Thus, $\beta^+_{ij}(z)$ (resp. $\beta^-_{ij}(z)$) is the maximum increase (resp. decrease) $b$ in $z_{ij}$ such that $z+b\cdot\e_{ij}$ (resp. $z-b\cdot\e_{ij}$) remains in $Z$; see \cref{fig:beta}.
\begin{figure}
    \begin{center}
        \begin{tikzpicture}[scale=0.7]
    \filldraw[fill=blue!10, draw=blue] (0,0) ellipse (3 and 2);
    \node at (2.5,1.5) {$Z$};
    \filldraw (0,0) circle (2pt);
    \node[below] at (0,0) {$z$};
    \draw[->, thick] (0,0) -- (3,0);
    \node[above] at (1.5,0) {$\beta^+_{ij}(z)$};
    \node[below] at (2.5,0) {$\e_{ij}$};

    \draw[->, thick] (0,0) -- (-3,0);
    \node[above] at (-1.5,0) {$\beta^-_{ij}(z)$};
    \node[below] at (-2.5,0) {$-\e_{ij}$};
\end{tikzpicture}
    \end{center}
    \caption{Definition of the functions $\beta^+_{ij}(z)$ and $\beta^-_{ij}(z)$. The values $\beta^+_{ij}(z)$ and $\beta^-_{ij}(z)$ are the distances from $z$ to the boundary of the convex set $Z$ along the directions $\e_{ij}$ and $-\e_{ij}$ respectively.}\label{fig:beta}
\end{figure}
We now prove that the functions $\beta^+_{ij}(z)$ and $\beta^-_{ij}(z)$ are continuous in $z$; in fact we show they are concave. The lemma below records the more general observation that the distance $d(x)$ from a point $x$ in interior of a convex set $X$ to the boundary of $X$ along a particular direction $\e$ is a concave function in $x$.

\begin{lemma}\label{lem:convex-func}
Given a compact convex set $C \in \mathbb{R}^n$, let $\e\in \mathbb{R}^n$ be a unit vector along a particular direction. Define a function $d:X \rightarrow \mathbb{R}$ given by $d(x) = \max \{\alpha \ge 0 \mid x + \alpha \cdot \e \in C\}$. Then $d$ is a concave function.
\end{lemma}
\begin{proof}
Consider two arbitrary points $x,y \in C$. For arbitrary $\lambda\in[0,1]$, we prove $d(\lambda x + (1-\lambda)y) \geq \lambda d(x) + (1-\lambda) d(y)$. By definition of $d$, there exists $x^*,y^* \in C$ such that $x^*-x = d(x) \cdot \e$ and $y^* - y = d(y) \cdot \e$. We have 
    \begin{align*}
        \big(\lambda x^* + (1-\lambda) y^*\big) - \big(\lambda x + (1-\lambda)y\big) &= \lambda(x^* - x) + (1-\lambda)(y^* - y) \\
        &= \big( \lambda d(x) + (1-\lambda) d(y) \big) \cdot \e.
    \end{align*}
By convexity of $C$, $\lambda x^* + (1-\lambda) y^* \in C$. Therefore, $d(\lambda x + (1-\lambda)y) \geq \lambda d(x) + (1-\lambda) d(y)$. Hence $d$ is concave.
\end{proof}
As a corollary, we obtain the continuity of the functions $\beta^+_{ij}$ and $\beta^-_{ij}$.
\begin{corollary}[of Lemma \ref{lem:convex-func}]
For any $i,j \in [n]$, $\beta^+_{ij}$ and $\beta^-_{ij}$ are concave and thus continuous functions.
\end{corollary}

\begin{remark}
\normalfont We remark that ensuring continuity of our ``\emph{increase/decrease when feasible scheme}'' via the functions $\beta^+_{ij}$ and $\beta^-_{ij}$ crucially requires the convexity of our domain $Z$. This implying that convexity of $Z$ plays a more significant role than merely satisfying the conditions required of the domain in Brouwer's fixed point theorem.
\end{remark}

\paragraph{Updating a non-reciprocal exchange.} Consider an element $z\in Z$ s.t. the exchange $f(z)$ is not reciprocal. Then there must be agents $i, j \in \N$ s.t. $\Delta_i(f(z)) \neq \Delta_j(f(z))$. For any two agents $i,j$, we design a function $g^{ij}$ which updates $z$ towards a element $z'$ s.t. $f(z')$ is more reciprocal than $f(z)$ for the agents $i$ and $j$. Hence, the function $g^{ij}$ should map $z$ to an element $z'$ where $z_{ij}$ is increased if $\Delta_i(f(z)) < \Delta_j(f(z))$ and decreased if $\Delta_i(f(z)) > \Delta_j(f(z))$. Moreover, we perform this update as long as $z'$ remains in $Z$. The extent of such updates is given by $\beta^+_{ij}(z)$ and $\beta^-_{ij}(z)$ by definition. Thus a natural choice of $g^{ij}$ is given by $g^{ij}(z) = z + \e_{ij}\cdot \beta^+_{ij}(z)$ if $\Delta_i(f(z)) < \Delta_j(f(z))$, and $g^{ij}(z) = z - \e_{ij}\cdot \beta^-_{ij}(z)$ if $\Delta_i(f(z)) > \Delta_j(f(z))$. However, this definition compromises continuity of $g^{ij}$ over $Z$, as it depends on the cases of $\Delta_i(f(z)) < \Delta_j(f(z))$ or $\Delta_i(f(z)) > \Delta_j(f(z))$. Thus, for continuity, we incorporate the cases into $g^{ij}$ by including the difference $(\Delta_i(f(z)) - \Delta_j(f(z)))$ as a multiplier. Let $\delta = 1/\max_i u_i(1^{n\times n})$. We define $g^{ij}: Z \rightarrow Z$ as follows:
\begin{equation}\label{eq:g-ij}
g^{ij}(z) = \begin{cases}
z + \e_{ij}\cdot\delta\cdot\beta^+_{ij}(z)\cdot (\Delta_j(f(z))-\Delta_i(f(z))), \text{ if } \Delta_i(f(z))\le \Delta_j(f(z)), \\
z - \e_{ij}\cdot\delta\cdot\beta^-_{ij}(z)\cdot (\Delta_i(f(z))-\Delta_j(f(z))), \text{ otherwise.}
\end{cases}
\end{equation}

\begin{figure}
    \begin{center}
        \begin{tikzpicture}[scale=0.9]
    \draw[->] (-0.5, 0) -- (6, 0) node[right] {$z_{12}$};
    \draw[->] (0, -0.5) -- (0, 4) node[above] {$z_{21}$};
    \filldraw[fill=blue!10, draw=blue] (3, 2) ellipse (2.5 and 1.5);
    \node at (4.5, 3.5) {$Z$};
    \filldraw (1.8, 2.4) circle (2pt);
    \node[above left] at (1.8, 2.4) {$z$};
    \draw[->, thick] (1.8, 2.4) -- (4.0, 2.4);
    \node[right] at (4.0, 2.4) {$g^{12}(z)$};

    \draw[->, thick] (1.8, 2.4) -- (1.8, 1.6);
    \node[below] at (1.8, 1.6) {$g^{21}(z)$};

    \draw[->, thick] (1.8, 2.4) -- (2.9, 2.0);
    \filldraw (2.9, 2.0) circle (2pt);
    \node[below right] at (2.8, 2.1) {$g(z)$};

\end{tikzpicture}
    \end{center}
    \caption{Illustrating the function $g(z)$. The point $z \in [0,M]^2$ is s.t. $\Delta_1(f(z)) < \Delta_2(f(z))$. Hence $g$ increases $z_{12}$ to $g^{12}(z)$ and decreases $z_{21}$ to $g^{21}(z)$, and returns $g(z)$ as a convex combination of $z$, $g^{12}(z)$, and $g^{21}(z)$.}
    \label{fig:g-func}
\end{figure}

Finally, we define $g(z)$ to combine the updates $g^{ij}(z)$ for all pairs of agents $i,j \in [n]$, as $g(z) = \sum_{i,j} \frac{1}{n^2} g^{ij}(z)$; see Figure~\ref{fig:g-func} for an illustration. We prove below that $g$ indeed maps elements of $Z$ to itself, which crucially uses the convexity of set $Z$.

\begin{lemma}\label{lem:g-func}
For all $z \in Z$, $g(z) \in Z$.
\end{lemma}
\begin{proof}
We prove that $g^{ij}(z) \in Z$ for all $i,j \in [n]$. By definition of $\delta$, we have $\delta\cdot |\Delta_j(f(z)) - \Delta_i(f(z))| \le 1$. Then by definition of $\beta^+_{ij}(z)$ and $\beta^-_{ij}$, we have that $g^{ij}(z) \in Z$. Finally, by convexity of $Z$ from \cref{obs:convexcompact}, we also have that $g(z) = \sum_{ij} \frac{1}{n^2} g^{ij}(z) \in Z$.
\end{proof}
Next, we prove that $g$ is continuous. 
\begin{lemma}\label{lem:continuous}
$g:Z \rightarrow Z$ is a continuous function.
\end{lemma}
\begin{proof}
Observe that it suffices to prove that for all $i,j \in [n]$, $g^{ij}$ is a continuous function. Fix $i, j\in [n]$. Define $h^{ij}: Z\rightarrow \mathbb{R}$ as $h^{ij}(z) = \Delta_j(f(z)) - \Delta_i(f(z))$. The continuity of $\psi_{kh}$ for all $k,h\in [n]$ and the continuity of $f$ imply that $\Delta_i(f(z))$ and $\Delta_j(f(z))$ are continuous in $z$. Thus, $h^{ij}$ is also continuous in $z$.

Let $\mathcal{B}(z, \rho)$ denote  an open ball of radius $\rho$ around $z$. We argue that $g^{ij}$ is continuous at every $z\in Z$, by showing that for every $\rho > 0$, there exists a $\gamma > 0$ s.t. $z' \in \mathcal{B}(z, \gamma)$ implies $g^{ij}(z') \in \mathcal{B}(g^{ij}(z), \rho)$. Fix a point $z\in Z$ and $\rho > 0$. We consider three cases:
\begin{itemize}
\item $h^{ij}(z) = 0$. Since $h^{ij}$ is continuous, there exists some $\gamma' > 0$ such that $z'\in \mathcal{B}(z, \gamma')$ implies $h^{ij}(z') \in \mathcal{B}(h^{ij}(z), \frac{\rho}{2\delta M}) = \mathcal{B}(0, \frac{\rho}{2\delta M})$. For $z'\in \mathcal{B}(z, \gamma')$, we define $\beta_{ij}(z') = \beta^+_{ij}(z')$ if $h^{ij}(z') \ge 0$, and $\beta^-_{ij}(z')$ otherwise. Using the definition of $g^{ij}$ from \cref{eq:g-ij}, we have for every $z'\in\mathcal{B}(z, \gamma')$:
\begin{align*}
|g^{ij}(z') - g^{ij}(z)| &\le |z'-z| + \delta\cdot|\beta_{ij}(z')\cdot h^{ij}(z')| \\
&\le |z'-z| + \delta\cdot M \cdot \frac{\rho}{2\delta M} \tag{using $|\beta^{ij}(z')| \le M$}\\
&\le |z'-z| + \frac{\rho}{2}.
\end{align*}
We set $\gamma := \min(\gamma', \rho/2)$. Then, the above inequality implies that for all $z'\in \mathcal{B}(z, \gamma)$, we have $g^{ij}(z') \in \mathcal{B}(g^{ij}(z), \rho)$. This shows that $g^{ij}$ is continuous at $z$.
\item $h^{ij}(z) = r > 0$. Since $h^{ij}$ is continuous, there exists some $\gamma' > 0$ such that $z' \in \mathcal{B}(z, \gamma)$ implies $h^{ij}(z') \in \mathcal{B}(h^{ij}(z), \frac{r}{2})$. Thus, $|h^{ij}(z') - h^{ij}(z)|  < r/2$. Since $h^{ij}(z) = r$, we conclude $h^{ij}(z') > 0$ for all $z'\in \mathcal{B}(z, \gamma')$. This implies that $g^{ij}(z') = z' + \e_{ij}\cdot \delta\cdot\beta^+_{ij}(z')\cdot h^{ij}(z')$ for all $z'\in \mathcal{B}(z, \gamma')$. 

Now consider the function $\bar{g}^{ij} : Z \rightarrow Z$ given by $\bar{g}^{ij}(z') := z' + \e_{ij}\cdot\delta\cdot\beta^+_{ij}(z')\cdot h^{ij}(z')$. We have that $g^{ij}(z') = \bar{g}^{ij}(z')$ for all $z'\in \mathcal{B}(z, \gamma')$. Further, since $\beta^+_{ij}$ and $h^{ij}$ are continuous functions, $\bar{g}^{ij}$ is also a continuous function. Thus, there exists some $\gamma'' > 0$ s.t. $z'\in\mathcal{B}(z, \gamma'')$ implies $\bar{g}^{ij}(z') \in \mathcal{B}(\bar{g}^{ij}(z), \rho)$. 

We set $\gamma := \min(\gamma', \gamma'')$. We obtain that $z' \in \mathcal{B}(z, \gamma)$ implies $g^{ij}(z') \in \mathcal{B}(g^{ij}(z), \rho)$, since $g^{ij}(z') = \bar{g}^{ij}(z')$ for all $z'\in \mathcal{B}(z, \gamma')$. This proves that $g^{ij}$ is continuous at $z$.
\item $h^{ij}(z) = r < 0$. An analysis analogous to the above case shows that $g^{ij}$ is continuous at $z$ in this case too.
\end{itemize}
In conclusion, we showed that for all $i,j\in \N$, $g^{ij}$ is continuous at all points $z\in Z$. Since $g : = \sum_{i, j\in [n]} \frac{1}{n^2} g^{ij}$, we conclude that $g$ is a continuous function as well.
\end{proof}

We can now prove \cref{thm:fair-cs-existence}.
\begin{proof}[Proof of \cref{thm:fair-cs-existence}] Recall that $Z$ is a non-empty, convex, compact set (\cref{obs:convexcompact}), and that the map $g$ is a continuous function from $Z$ to $Z$ in Lemmas~\ref{lem:g-func} and \ref{lem:continuous}. By Brouwer's fixed-point theorem, $g$ has a fixed-point $z^*\in Z$.

Let $\x=f(z^*)$. We show that $\x$ is reciprocal, taking into account that the functions $\beta^+_{ij}$ and $\beta^-_{ij}$ may vanish on the boundary of $Z$. Since $g^{ab}$ changes only coordinate $(a,b)$, we have, for every $a,b\in\N$,
$g(z^*)_{ab}-z^*_{ab}
=\frac{1}{n^2}\bigl(g^{ab}(z^*)_{ab}-z^*_{ab}\bigr)=0.$

Suppose, for a contradiction, that $\x$ is not reciprocal, and let $S=\{i\in\N:\Delta_i(\x)>0\}$. Since $\sum_i\Delta_i(\x)=0$, the set $S$ is non-empty and proper. By efficiency of the share functions, the contributions within $S$ cancel, implying
\[
0<\sum_{i\in S}\Delta_i(\x)
=\sum_{\substack{i\in S\\j\notin S}}\psi_{ij}(\x_j)
-\sum_{\substack{i\notin S\\j\in S}}\psi_{ij}(\x_j).
\]
All shares are non-negative, so there exist $i\in S$ and $j\notin S$ with $\psi_{ij}(\x_j)>0$. Normalization then implies $x_{ij}>0$, and hence $z^*_{ij}>0$. Moreover, $\Delta_i(\x)>\Delta_j(\x)$.

If $\beta^-_{ij}(z^*)>0$, then \cref{eq:g-ij} gives $g^{ij}(z^*)_{ij}<z^*_{ij}$, contradicting the fixed-point identity above. We must therefore have $\beta^-_{ij}(z^*)=0$. Write $B=n\log(1/\alpha)$, so that $M>B$. Since $z^*_{ij}>0$, the lower bound on this coordinate does not prevent a small decrease. The finitely many cycle constraints defining $Z$ must therefore include a tight constraint for a directed cycle $C$ containing $(i,j)$:
$\sum_{(a,b)\in C}z^*_{ab}=B.$
Along the remaining directed path from $j$ to $i$ on $C$, the surplus starts at $\Delta_j(\x)$ and ends at the strictly larger value $\Delta_i(\x)$. Thus some edge $(a,b)$ on this path satisfies $\Delta_a(\x)<\Delta_b(\x)$. Since all coordinates of $z^*$ are non-negative, tightness of $C$ implies $z^*_{ab}\le B<M$. Increasing $z^*_{ab}$ preserves every cycle inequality, so $\beta^+_{ab}(z^*)=M-z^*_{ab}>0$. By \cref{eq:g-ij}, we obtain $g^{ab}(z^*)_{ab}>z^*_{ab}$, again contradicting the fixed-point identity. Therefore, $\x=f(z^*)$ is reciprocal.

Finally, \cref{lem:z-cs} implies that the exchange $f(z^*)$ is also $\eps$-\cs.
\end{proof}

\subsection{Existence of reciprocal and exact core-stable exchanges}

Having proved the existence of reciprocal and $\eps$-\cs exchanges for any $\eps \in (0, 1)$, we next prove the existence of reciprocal and exactly core-stable exchanges. Our proof relies on sequential compactness of the sets of reciprocal and $\eps$-\cs exchanges. We prove this in the next two lemmas by using the continuity of utility and sharing functions.

\begin{lemma}\label{lem:reciprocal-compact}
The set of reciprocal exchanges is compact.
\end{lemma}
\begin{proof}
For $i\in \N$, let $\hat{R}_i^+ = \{\x \in [0,1]^{n\times n} : \psi_i(\x) > u_i(\x)\}$ and $\hat{R}_i^- = \{\x \in [0,1]^{n\times n} : \psi_i(\x) < u_i(\x)\}$. By definition, $\hat{R}_i^+ \cup \hat{R}_i^-$ is the set of exchanges that are not reciprocal for agent $i$. By continuity of $u_i$ and $\psi_i$, we see that both $\hat{R}_i^+$ and $\hat{R}_i^-$ are open sets. Hence the set $\bar{R} = \bigcup_{i\in \N} \hat{R}_i^+ \cup \hat{R}_i^-$ is also an open set. Observe that $\bar{R}$ is exactly the set of exchanges that are not reciprocal. Hence, the set of reciprocal exchanges is closed. Since the set of exchanges is also bounded, the compactness of the set of reciprocal exchanges follows.
\end{proof}

\begin{lemma}\label{lem:eps-cs-compact}
The set of $\eps$-\cs exchanges is compact, for any $\eps \in [0,1]$.
\end{lemma}
\begin{proof}
For a subset $\N'$ of agents and an agent $i\in \N'$, consider the set $S(\N', i)$ of exchanges in which agent $i$ will gain significantly if agents in $\N'$ fully exchange data among themselves. Let $\y[\N'] \in [0,1]^{n\times n}$ be an exchange where $\y[\N']_{kh} = 1$ for $k,h\in \N'$ and $0$ otherwise. Then $S(\N', i) = \{ \x \in [0,1]^{n\times n} : u_i(\y[\N']) > u_i(\x) + \eps \}$. Note that by continuity of $u_i$, the set $S(\N', i)$ is an open set for all $\N'\subseteq \N$ and $i\in \N'$. 

The definition of $S(\N', i)$ and $\eps$-core-stability implies that the set of exchanges which are \textit{not} $\eps$-\cs is given by $\bigcup_{\N'\subseteq \N} \bigcap_{i\in \N'} S(\N', i)$. Since the intersection of finitely many open sets, and the union of open sets is open, we obtain that the set of exchanges which are not $\eps$-\cs is open. Thus, the set of $\eps$-\cs exchanges is closed. Since the set of exchanges is also bounded, the compactness of the set of $\eps$-\cs exchanges follows.
\end{proof}

We now show the existence of reciprocal and exact \cs exchanges.

\begin{theorem}\label{thm:fair-cs-existence-exact}
Any data exchange instance $\I$ where utility functions and sharing functions are continuous and monotone admits a reciprocal and \cs exchange.
\end{theorem}
\begin{proof}
\cref{thm:fair-cs-existence} implies the existence of reciprocal and $\eps$-\cs exchanges for any $\eps \in (0, 1)$. Thus, for any $k\in \mathbb{N}$, there exists an exchange $\x^k$ that is reciprocal and $\frac{1}{k}$-\cs. Since the set of reciprocal exchanges is compact (\cref{lem:reciprocal-compact}), the Bolzano-Weirerstrass theorem implies that the sequence $\{\x^k\}_{k\in\mathbb{N}}$ has a convergent subsequence $\{\x^{\sigma_i}\}_{i\in\mathbb{N}}$, which converges to a reciprocal exchange $\x^*$. 

We argue that $\x^*$ is also \cs. Fix an $\ell\in \mathbb{N}$. Let $i_\ell \in \mathbb{N}$ be the smallest integer $i$ s.t. $\sigma_i \ge \ell$. In other words, $\x^{i_\ell}$ is smallest-index element of the subsequence $\{\x^{\sigma_i}\}_{i\in\mathbb{N}}$ which appears no earlier than $\x^\ell$ in the sequence $\{\x^k\}_{k\in\mathbb{N}}$. Then, the sequence $\{\x^{\sigma_i}\}_{i\ge i_\ell}$ is a subsequence of the convergent subsequence $\{\x^{\sigma_i}\}_{i\in \mathbb{N}}$, and hence it is also convergent with the same limit point $\x^*$.

Moreover, since an $\eps$-\cs exchange is also $\eps'$-\cs for $\eps' > \eps$, $\{\x^{\sigma_i}\}_{i\ge i_\ell}$ is a sequence of $\frac{1}{\ell}$-\cs exchanges. Since the set of $\frac{1}{\ell}$-\cs exchanges is compact (\cref{lem:eps-cs-compact}), we conclude that the limit point $\x^*$ is also $\frac{1}{\ell}$-\cs. Thus, $\x^*$ is $\frac{1}{\ell}$-\cs for every $\ell\in\mathbb{N}$. Suppose $\x^*$ is not exactly \cs. Then, there exists some $k\in\mathbb{N}$ such that $\x^*$ is not $\frac{1}{k}$-\cs, which is a contradiction. Hence, $\x^*$ is both reciprocal and exactly \cs.
\end{proof}

\section{Computational Results}\label{sec:comp-main}
In this section, we investigate the complexity of computing reciprocal and core-stable exchanges, having shown existence in \cref{sec:existence}. We first note that modifying the data contribution of any single agent (increasing/decreasing $x_{ij}$), may potentially change the utility contributions ($\psi_{i'j}(\x_j)$ for $i' \neq i$) of all agents. Designing an algorithm with no assumption on the structure of these changes seems very challenging -- in fact, we suspect that in its most generality, computing an $\varepsilon$-core-stable and $\varepsilon$-reciprocal exchange is PPAD-complete. 

We therefore make two minimal assumptions: (i) the utility functions are $L$-Lipschitz continuous, and (ii) the share functions are \emph{cross-monotone}, i.e., increasing $x_{ij}$ can only decrease $\psi_{i'j}(\x_j)$ for all $i' \neq i$. Cross-monotonicity captures instances exhibiting diminishing marginal gains with data -- in particular, when $u_i(\cdot)$ is continuous monotone submodular for all $i$, and $\psi_{ij}(\cdot)$ is the Shapley share of $i$'s data in $u_j(\x_j)$, then $\psi_{ij}(\cdot)$ satisfies cross-monotonicity. For our computational results, we assume oracle access to the utilities and share functions, i.e., given exchange $\x$, we assume oracle access to $u_i(\x_i)$ and $\psi_{ij}(\x_j)$ for all $i, j \in \N$. 
Our main result is that:

\begin{restatable}{theorem}{thmCLS}\label{thm:CLS}
Determining a $\varepsilon$-reciprocal and $\varepsilon$-core stable exchange is in CLS when $L/\varepsilon = \textup{poly}(n)$.    
\end{restatable}

Since \cite{FearnleyGHS23} showed that the complexity class CLS is in fact equal to PPAD $\cap$ PLS, we prove membership in PPAD and in PLS separately to prove \cref{thm:CLS}. 

\paragraph{PPAD Membership.} We provide an overview of the proof of PPAD-membership, and defer the details to \cref{sec:comp}. \cite{etessami2010fixedpoint} introduce sufficient conditions for a fixed point problem to be in PPAD. One of these conditions requires the domain to be a polytope with polynomially many linear constraints. Our current formulation $f \colon Z \rightarrow Z$ does not satisfy this condition. We bypass this problem by embedding our domain $Z$ in a bigger polytope $Z' \supseteq Z$ with polynomially many linear inequalities, and defining a $\tilde{g} \colon Z' \rightarrow Z'$, such that $\tilde{g} (z) = f(\pi(z))$, where $\pi(z)$ is the Euclidean-projection of $z$ onto $Z'$. This projection ensures that all fixed points of $\tilde{g}(\cdot)$ in $Z'$ correspond to fixed points of $f(\cdot)$ in $Z$. Further, it can be argued that this projection is unique, rational, and \emph{polynomial time computable}. Polynomial time computability crucially exploits the facts that (i) while our original domain $Z'$ may have exponentially many constraints, it still admits a poly-time separation oracle, and (ii) the projection from $z \notin Z$ onto $Z$ is rational and has polynomial bit length. This puts the problem of determining a $\varepsilon$-core-stable and $\varepsilon$-reciprocal exchange in PPAD.  

\subsection{Local Search Algorithm}
We outline a \emph{local-search} algorithm to compute a $\varepsilon$-core-stable and $\varepsilon$-fair exchange for instances with $L$-Lipshcitz continuous utility functions and cross-monotone shares. We also use the algorithm to show membership in PLS. Below, we briefly outline our key ideas behind the local search procedure, and defer the full details to \cref{sec:comp}.

At a high level, our algorithm maintains the invariant that $G(\x, \varepsilon/(nL))$ is acyclic and reduces the surplus of the agents with high surpluses. Under the $L$-Lipschitzness assumption on the utility functions, it can be shown that every $\x$ such that $G(\x, \varepsilon/(nL))$ is acyclic is $\varepsilon$-core stable. Given an exchange $\x$, define $\Delta(\x) = \langle \Delta_{\sigma(1)}(\x), \Delta_{\sigma(2)}(\x), \dots,$ $ \Delta_{\sigma(n)}(\x) \rangle$, where $\sigma \colon \N \rightarrow \N$ sorts the agents according to the surplus levels, i.e., $\Delta_{\sigma(1)}(\x) \geq \Delta_{\sigma(2)}(\x) \geq  \dots \geq \Delta_{\sigma(n)}(\x)$. Our algorithm maintains the invariant that $\x \in G(\x,\varepsilon/(nL))$, and decreases $\Delta(\x)$ lexicographically in every iteration. The algorithm terminates when $\max_{i \in \N} \Delta_i(\x) \leq \varepsilon/n$, as this will ensure that no positive surplus is larger than $\varepsilon/n$, and no negative surplus is smaller than $\varepsilon$: since $\sum_i \Delta_i(\x) = 0$, we have $\Delta_j(\x) \geq -n \max_i \Delta_i(\x)$, implying that $\x$ is $\varepsilon$-fair.

Given an exchange $\x$, with $\max_{i \in \N} \Delta_i(\x) > \varepsilon/n$,  we partition the agents into two sets based on their surpluses by the following procedure: Let $S = \{i_h\}$, where $i_h$ is the agent with the highest surplus. Iteratively, select the agent in $\N \setminus S$ with the highest surplus, say $i$, and add $i$ to $S$ only if the surplus of $i$ is at most $\varepsilon/n^2$ less than the lowest surplus in $S$, i.e., $\Delta_i(\x) +\varepsilon/n^2 \geq \min_{j \in S} \Delta_j(\x)$. Observe that at termination, $S$ will comprise only positive surplus agents as the lowest surplus in $S$ is at least $\max_{i \in \N} \Delta_i(\x) - |S| \cdot (\varepsilon/n^2) > 0$. Intuitively, $S$ is the set of agents with high surpluses (they give more than they get). Also note that we have $\min_{i \in S} \Delta_i(\x) > \max_{j \in \N \setminus S} \Delta_j(\x) + \varepsilon/n^2$. The algorithm's goal would be to reduce the surplus of agents in $S$ (with at least one agent experiencing a strict reduction), and ensure that no surplus in $\N \setminus S$ becomes more than the surplus of some agent in $S$, thereby lexicographically reducing $\Delta(\x)$. This is achieved (i) by carefully reducing the data flow from $S$ to $\N \setminus S$, or (ii) by increasing data flow from $\N \setminus S$ to $S$. We discuss the two cases below.

\paragraph{Decreasing Data Flow from $S$ to $\N \setminus S$:} This case is nuanced, as one needs to be careful about maintaining the invariant of the algorithm. Decreasing data flow from an agent $i \in S$ to an agent $j \in \N \setminus S$, when there is a path from $j$ to $i$ in $G(\x, \varepsilon/(nL))$, can violate the invariant that $G(\x, \varepsilon/(nL))$ is acyclic. So we will only decrease flow when there are no edges in $G(\x,\varepsilon/(nL))$ from $\N \setminus S$ to $S$.

Since agents in $S$ have high surplus (their utility outflow is larger than the inflow) and agents in $\N \setminus S$ (their utility inflow is larger than the outflow) have lower surplus, we show that there is an agent $j \in \N \setminus S$, that receives substantial data flow from $S$-- in particular, we show that there is an agent $j \in \N \setminus S$ such that $\sum_{i \in S} \psi_{ij}(\x_j) > \varepsilon/n^2$. Our goal would be to reduce the surpluses of the agents in $S$ by reducing their data contributions to $j$.

Let $S(j) \subseteq S$ be the agents such that $x_{ij} = 0$ for all $i \notin S(j)$. Pick any agent $\ell \in S(j)$. Note that when we decrease $x_{\ell j}$, $\psi_{\ell j}(\x_j)$ decreases, and $\psi_{\ell'j}(\x_j)$ is non-decreasing for all $\ell' \neq \ell$ (due to cross-monotonicity). Also note that $u_j(\x_j)$ increases and $u_{j'}(\x_{j'})$ remains unchanged for all $j' \neq j$, implying that $\Delta_{\ell}(\x)$ decreases, and $\Delta_{\ell'}(\x)$  for all $\ell' \neq \ell$ (including $j$) is non-increasing. This could potentially increase the surplus of a high-surplus agent in $S$, inadvertently causing a lexicographic increase in $\Delta(\x)$ (see Figure~\ref{fig:algorithm-case1})! We fix the foregoing issue by a \emph{threshold based data reduction} from $S$ to $\N \setminus S$.

\begin{figure}
    \begin{center}
        \begin{tikzpicture}[scale=0.5,
roundnode/.style={circle, draw=green!60, fill=green!5, very thick, minimum size=7mm},
squarednode/.style={rectangle, draw=red!60, fill=red!5, very thick, minimum size=5mm},
]
\draw[black, fill=green!5, very thick] (0,0) rectangle (4.5,4);
\draw[black, fill=blue!5, very thick]  (6,-3.5) rectangle (12, 1);
\node at (2.5,-0.4) {$S$};
\node at (9.1,-3.1) {$\N \setminus S$};

{

\draw[black, fill=blue!15, very thick] (1.5,0) rectangle (3,3);
\draw[black, fill=blue!15, very thick] (3,0) rectangle (4.5,2);

\draw[black, fill=blue!15, very thick] (6,0) rectangle (7.5,-2);
\draw[black, fill=blue!15, very thick] (7.5,0) rectangle (9,-1);
\draw[black, fill=blue!15, very thick] (9,0) rectangle (10.5,-2.7);
\draw[black, fill=blue!15, very thick] (10.5,0) rectangle (12,0.5);

}

{
\draw[black, fill=blue!15, very thick] (0,0) rectangle (1.5,2.5);

}

\draw[->, blue, ultra thick] (5,-5)--(5,-8);
\node at (6,-7)  {{\textcolor{blue}{$x_{ij} \downarrow$}}};

\draw[black, fill=green!5, very thick] (0,0-13) rectangle (4.5,4-13);
\draw[black, fill=blue!5, very thick]  (6,-3.5-13) rectangle (12, 1-13);
\node at (2.5,-0.4-13) {$S$};
\node at (9.1,-3.1-13) {$\N \setminus S$};

\node at (0.75,1) {$i$};
\node at (6.75,-0.9) {${\textcolor{black}{j}}$};
\node at (8.25,-0.4) {${\textcolor{black}{k}}$};
\node at (11.25,0.2) {${\textcolor{black}{\ell}}$};
\node at (2.25,1) {${\textcolor{black}{\tilde{i}}}$};

\node at (0.75,1-13) {$i$};

\draw[black, fill=blue!15, very thick] (0,0-13) rectangle (1.5,2-13);
\draw[red, dashed, very thick] (0,2.5-13)--(1.5,2.5-13);
\node at (0.75,1-13) {$i$};

\draw[black, fill=blue!15, very thick] (1.5,0-13) rectangle (3,3.25-13);
\draw[black, fill=blue!15, very thick] (3,0-13) rectangle (4.5,2-13);
\draw[red, dashed, very thick] (1.5,3-13)--(3,3-13);

\draw[black, fill=blue!15, very thick] (6,0-13) rectangle (7.5,-1.8-13);
\draw[black, fill=blue!15, very thick] (7.5,0-13) rectangle (9,-0.8-13);
\draw[black, fill=blue!15, very thick] (9,0-13) rectangle (10.5,-2.7-13);
\draw[red, dashed, very thick] (6,-2-13)--(7.5,-2-13);
\draw[red, dashed, very thick] (7.5,-1-13)--(9,-1-13);

\draw[black, fill=blue!15, very thick] (10.5,0-13) rectangle (12,0.75-13);
\draw[red, dashed, very thick] (10.5,0.5-13)--(12,0.5-13);

\node at (10,2.7-13){\textcolor{blue}{\textbf{Good:} $\Delta_i(\x) \downarrow$}};

\node at (10,2-13) {\textbf{\textcolor{red}{Problem: $\Delta_{\tilde{i}}(\x) \uparrow$}}};
\node at (10,3.4-13) {\textbf{\textcolor{blue}{Good: $\Delta_j(\x), \Delta_k(\x), \Delta_{\ell}(\x) \uparrow$}}};

\node at (6.75,-0.9-13) {${\textcolor{black}{j}}$};
\node at (8.25,-0.4-13) {${\textcolor{black}{k}}$};
\node at (11.25,0.2-13) {${\textcolor{black}{\ell}}$};
\node at (2.25,1-13) {${\textcolor{black}{\tilde{i}}}$};

\end{tikzpicture}
    \end{center}
    \caption{Illustration of the main bottleneck when we decrease data flow from $S$ to $\N\setminus S$. The top shows the surplus profiles before we reduced data flow from $i \in S$ to $j \in \N \setminus S$. The bottom shows the surplus levels after the decrease. The red-dashed lines in the histogram shows the surplus level prior to the change. The surplus of all agents in $\N \setminus S$ cannot decrease (in fact surpluses of agents $j,k$ and $\ell$ strictly increase) and the surplus of agent $i$ decreases. Both the foregoing changes are ideal as they balance the surplus profiles. Unfortunately, the surplus of another agent $\tilde{i} \in S$ can also increase, as $\psi_{\tilde{i}j}$ can increase as we decrease data flow from $i$ to $j$.}
    \label{fig:algorithm-case1}
    \end{figure}

\noindent \emph{Threshold Based Data Reduction from $S$ to $\N \setminus S$:} For each agent $i$, we define a lower threshold limit $\delta_i = \Delta_i(\x) - \varepsilon/2n^3$ for her surplus. Throughout the data reduction procedure, we will never let the surplus of agent $i$ drop below $\delta_i$, i.e.,  we perform reductions in $x_{ij}$ until either $x_{ij}$ becomes $0$, or surplus of agent $i$ drops by at least $\varepsilon/4n^3$ (from the beginning). We now describe the data reduction: \emph{while there are agents $i \in S$ with $\Delta_i(\x) - \delta_i > \varepsilon/4n^3$ and $x_{ij} > 0$, then reduce $x_{ij}$ until either $x_{ij} = 0$ or $\Delta_i(\x) = \delta_i$.} Note that while an agent $i$'s surplus may be reduced to $\delta_i$ in an iteration, it can again increase in the next iteration when $x_{i'j}$ is reduced (for some $i' \neq i$) and as a result $\psi_{ij}(\x_j)$ increases. However, we show that every iteration will involve an $x_{ij} = 0$ or some $x_{ij}$ being reduced by $\varepsilon/(4n^3L)$ (follows from the $L$-Lipschitzness of the utility). Also, throughout the algorithm, the data contributions from $S$ to $\N \setminus S$ only decrease. Therefore, the number of iterations is $\textup{poly}(n,L,1/\varepsilon)$. Now, observe crucially that at the end of the iteration:
\begin{enumerate}[label=(\roman*)]
    \item \emph{No surplus in $S$ will increase:}  for any agent $i \in S$, if $x_{ij} = 0$, then $\psi_{ij}(\x_j) = 0$, implying that its contribution (and consequently its surplus $\Delta_i(\x)$) has not increased. Therefore, if there is an agent $\tilde{i} \in S$, that has larger surplus than before, then $x_{\tilde{i}j} > 0$ and since $\Delta_{\tilde{i}}(\x) - \delta_{\tilde{i}} > \varepsilon/2n^3 > \varepsilon/4n^3$, the algorithm has not terminated.

    \item \emph{The surplus of at least one agent in $S$ will reduce by $\varepsilon/4n^3$:} at the beginning of the procedure, there exists one agent $i \in S$ such that $\psi_{ij}(\x_j) > \varepsilon/n^3$ (as $\sum_{i \in S} \psi_{ij}(\x_j) > \varepsilon/n^2$). At the end of this procedure, we have either (a) $\Delta_i(\x) - \delta_i < \varepsilon/4n^3$, in which case $\Delta_i(\x)$ has reduced by at least $\varepsilon/4n^3$, or (b) $x_{ij}= 0$, in which case, $\Delta_i(\x)$ will reduce by $\varepsilon/n^3 > \varepsilon/4n^3$.
    
    \item \emph{No agent in $\N \setminus S$ will have higher surplus than any agent in $S$:} since the sum of surpluses is always zero, the total decrease in the surpluses of the agents in $S$, equals the total increase in the surplus of the agents in $\N \setminus S$. Since, no agent's surplus in $S$ decreases by more than $\varepsilon/2n^3$, the maximum increase in the surplus of any agent $j$ in $\N \setminus S$ is $n \cdot \varepsilon/2n^3 = \varepsilon/2n^2$. Since there was an initial gap of $\varepsilon/n^2$ between any surplus in $S$ and the surplus of $j$, the surplus of $j$ will still be lower than the surplus of any agent in $S$.
\end{enumerate}
 
Points (i), (ii), and (iii) show that threshold based data reduction will lexicographically reduce $\Delta(\x)$ while still maintaining $\varepsilon$-core stability. 

\paragraph{Increasing Data Flow from $\N \setminus S$ to $S$:} We now consider the case when there are edges from $\N \setminus S$ to $S$ in $G(\x,\varepsilon/(nL))$. This would make it infeasible to reduce flow from $S$ to $\N \setminus S$. We resort to increasing flow along the edges from $\N \setminus S$ to $S$. Consider an agent $j \in \N \setminus S$ and $i \in S$ such that ${(j,i)} \in G(\x, \varepsilon/(nL))$. Since ${(j,i)} \in G(\x, \varepsilon/(nL))$, we have $x_{ij} < 1 - \varepsilon/(nL)$. Observe that increasing $x_{ij}$ will increase $\psi_{ji}(\x_i)$, and decrease $\psi_{j'i}(\x_i)$ for all $j' \neq j$. Additionally, $u_i(\x_i)$ will decrease, while the utility of other agents remains the same. Therefore, increasing $x_{ij}$ will increase $\Delta_j(\x)$ and will decrease or leave unchanged  $\Delta_{j'}(\x)$ for all $j' \neq j$. Our algorithm increases $x_{ji}$ by $\varepsilon/n^3L$. This way, we argue that the surplus of at least one agent in $S$ decreases by $\varepsilon^2/n^4L$, and the surpluses of the agents in $\N \setminus S$ remain at or below the lowest surplus in $S$. Also, given that the surplus of no other agent in $S$ strictly increases, $\Delta(\x)$ decreases lexicographically. 

We note that the same local search algorithm can be used to show that determining an $\eps$-reciprocal and $\eps$-\cs exchange is in PLS when $L/\eps = \poly(n)$.
\section{Discussion}\label{sec:discussion}
In this paper, we introduce a formal model of data exchange. We outline our desiderata, drawing parallels to established theories in classical exchange economies. We show the existence of an exchange that satisfies the desiderata under minimal assumptions. Finally, we outline a local search algorithm for instances where agent utilities exhibit diminishing marginal gains with data. Under mild assumptions, we show that our problem is in CLS. We view this work as a first step in the theoretical exploration of data exchange economies, highlighting several promising avenues for further research. We highlight significant directions that emerge from our current results, as well as mention intriguing possibilities for generalizing our existing framework.

\paragraph{Computational Complexity of Finding a $\varepsilon$-Core-Stable and $\varepsilon$-Fair Exchange under Monotone Utilities and Cross-Monotone Share Functions.} While we show membership in CLS, the complexity of the problem is still open, even when $L/\varepsilon = \textup{poly}(n)$. We outline the main problem that arises when we try to generalize our local search algorithm. This will also highlight the main technical non-triviality in the next major direction. 

Observe that our local search algorithm attempts to ``balance'' the surpluses. When, we decrease the flow from $S$ to $\N \setminus S$ (when there are no edges from $\N \setminus S$ to $S$ in $G(\x, \varepsilon)$), note that surpluses in $S$ decrease or remain unchanged and surpluses in $\N \setminus S$ increase or remain unchanged. This creates an optimal balancing scenario: high surpluses in $S$  decrease, while low surpluses in $ \N \setminus S$ increase. Under these conditions, we can leverage more sophisticated balancing potential functions, such as the Euclidean norm of the surpluses $\sum_i \Delta^2_i(\x)$, which is polynomially bounded, and improves by an additive factor of $\textup{poly}(\varepsilon, 1/n, 1/L)$ every time we decrease flow from $S$ to $\N \setminus S$.

One would hope to adopt a symmetric approach while increasing flow from $\N \setminus S$ to $S$ (when decreasing flow from $S$ to $\N \setminus S$ is not feasible). However, this method encounters a major bottleneck:  note that increasing a flow from $j \in \N \setminus S$ to $i \in S$ causes the surplus of $j$ to increase (which is good), and all other surpluses to decrease or remain unchanged. While it is ideal for surpluses in $S$ to decrease, it could be problematic if surpluses in $\N \setminus S$ (low surplus) decrease. A natural solution would be to boost the outflow from these agents in $\N \setminus S$ whose surpluses have dropped to $S$. \emph{Unfortunately, the outflows of these agents may already be at maximum capacity ($x_{ji} = 1$ for all $i \in S$), making this approach infeasible! This situation contrasts sharply with reducing flow from  $S$ to $\N \setminus S$, where if an agent $i \in S$ is at minimum capacity (with $x_{ij} = 0$ for all $j \in \N \setminus S$), its surplus cannot have increased.}

We suspect that without the assumption that $L/\varepsilon = \textup{poly}(n)$, the problem is probably not in CLS, but it is also another question that demands a thorough investigation. Thus, the computational complexity of determining a $\varepsilon$-core-stable and $\varepsilon$-fair exchange is a deeply intriguing direction with many interesting open problems. 

\vspace{-0.15cm}

\noindent \paragraph{Computational Complexity of Finding a $\varepsilon$-Core-Stable and $\varepsilon$-Fair Exchange under Supermodular utilities.}
All our computational results discussed so far, apply to instances where agents experience diminishing marginal gains with data. However, there is another interesting subclass of instances, where the datasets are \emph{complimentary}, i.e., the utility gained from a combination of datasets is more than the sum of utilities gained from the individual datasets-- for example when data is divided by features~\cite{liu2024vertical}, it has been observed that combining datasets turn out to be significantly more valuable~\cite{castro2023data}. In particular, these are instances where  increasing $x_{ij}$ can increase or leave unchanged $\psi_{i'j}(\x_j)$ for all $i' \neq i$. Exploring the computational complexity of these complementary instances presents a promising avenue for future research.

We highlight the challenges encountered when adapting our local search algorithm. Decreasing data flow from an agent $i \in S$ to an agent in $j \in \N \setminus S$ can only cause the surplus of all agents other than $j$ to decrease (non-increase). This behavior aligns well with our potential function, as the high surpluses are decreasing. However, while increasing surplus from $j \in \N \setminus S$ to $i \in S$, the surpluses of all agents other than $i$ increase, including the high-surplus agents in $S$.

\emph{Interestingly, the core issue in both scenarios appears to stem from the need to implement the ``\emph{increase-flow}'' operation carefully-- transferring data from low-surplus agents to high-surplus agents.}

\vspace{-0.15cm}

\paragraph{Generalizing Current Framework.} Building on the framework outlined in the paper, a wide range of intriguing problems can be explored. Our current model relies on a single trusted central server that facilitates data exchange and allows for the retraining and refinement of machine learning models using data samples collected from multiple agents. However, the concept of \emph{decentralized data exchange} presents a fascinating avenue for investigation, where independent agents attempt to share data without a central authority. This shift opens up several compelling questions: (i) How do agents assess their utility and contribution functions for various data bundles? (ii) What forms of communication are permissible among the agents? (iii) Are there inherent dynamics that could guide us towards a desired solution? 
Each of these questions offers significant potential for foundational research in this emerging area.

\section*{Acknowledgments}
We thank Stepan (Styopa) Zharkov for the proof of \cref{thm:fair-cs-existence-exact}, suggesting the use of sequential compactness arguments together with \cref{thm:fair-cs-existence}.

\bibliographystyle{plainnat}
\bibliography{ref}

\clearpage
\appendix
\section{Missing Proofs from \texorpdfstring{\cref{sec:prelim}}{Section \getrefnumber{sec:prelim}}}\label{app:prelim}
\propshap*
\begin{proof}
Consider $\sum_{i\in \N} \psi_{ij}(\x_j)$. Fix a subset $\emptyset \neq S \subsetneq \N$. Let us examine how many times $u_j(\alloc_j[S])$ occurs as either the first or the second term of the difference in \cref{eq:shapley-share}. We observe that $u_j(\alloc_j[S])$ occurs as the first term exactly $|S|$ times, once for each $i\in S$, and as the second term exactly $(n-|S|)$ times, once for each $i\notin S$. Thus, the coefficient of $u_j(\x_j[S])$ in $\sum_{i\in\N} \psi_{ij}(\x)$ is $\frac{1}{n}\cdot \left(|S|\cdot  \binom{n-1}{|S|-1}^{-1} - (n-|S|)\cdot \binom{n-1}{|S|}^{-1} \right)$, which equals zero. Next, we note that the coefficient of $u_j(\x_j[\emptyset])$ is also zero, since $u_j(\cdot)$ is normalized. Finally, note that $u_j(\alloc_j[\N])$ appears only as the first term in the difference, and exactly $n$ times. Hence, its coefficient is $\frac{1}{n}\cdot (n-1)\cdot \binom{n-1}{n-1}^{-1} = 1$. In conclusion, $\sum_{i\in \N} \psi_{ij}(\x_j) = u_j(\x_j[\N]) = u_j(\x_j)$.
\end{proof}
\section{Computational Results}\label{sec:comp}
In this section, we focus on the computational aspects of determining a $\varepsilon$-reciprocal and $\varepsilon$-core-stable exchange. It is difficult to make statements on the computational complexity of the problem without making any assumptions on the utilities or the shares. Therefore, we assume that there is some constant $L > 1$ such that:
\begin{itemize}
\item The utility functions are $L$-Lipschitz continuous, i.e., for all $i,j\in \N$, exchanges $\x\in [0,1]^{n\times n}$, and $a\in[-1, 1]$, we have $|u_j(\x + a\cdot\e_{ij}) - u_j(\x)| \le L\cdot |a|$.
\item The share functions are cross-monotone, i.e., for all $i,j$, i.e., $\psi_{ij}(\cdot)$ is monotone non-increasing in $x_{i'j}$ for all $i' \neq i$, and $\pdv{\psi_{ij}}{x_{i'j'}} \leq L$ for all $i',j'$.
\end{itemize}
In particular, these instances exhibit diminishing marginal gains with data, e.g., we capture all instances where $u_i(\cdot)$s are $L$-Lipschitz, continuous monotone submodular, and $\psi_{ij}(x)$ is the Shapley share of $x_{ij}$ in $u_j(x_j)$. We can make the following stronger observations under these assumptions.
\begin{lemma}\label{lem:stronger-assumptions}
Let $\I$ be a data exchange instance with $L$-Lipschitz continuous utility and sharing functions. Then:
\begin{enumerate}
\item[(i)] $\delta = \frac{1}{\max_i u_i(1^{n})} \ge \frac{1}{n L}$.
\item[(ii)] $\alpha(\eps) \ge \frac{\eps}{nL}$ for $\eps \in (0,1)$.
\item[(iii)] For any exchange $\x\in[0,1]^{n\times n}$, if $G(\x, \frac{\eps}{nL})$ is acyclic, then $\x$ is $\eps$-\cs.
\end{enumerate}
\end{lemma}
\begin{proof} 
For part (i), we note that $|u_i(1^{n}) - u_i(0^{n})| \le n L$ using the Lipschitz continuity of $u_i$ and since each $u_i$ is a function of only $x_{ji}$ for $j\in \N$. Next, recall that the utility functions $u_i$ are normalized, i.e., $u_i(0^{n\times n}) = 0$ for all $i\in \N$. Thus, we observe that $\delta = (\max_i u_i(1^{n\times n}))^{-1} \ge \frac{1}{nL}$.

For part (ii), recall from \cref{def:alpha-eps} that $a_{ji}^\eps(\x) = \max\{ a\in [0,1] : u_i(\x - a\cdot \e_{ji}) \ge u_i(\x) - \eps/n\}$, for any $i,j\in[n]$ and exchange $\x$ with $x_{ji} = 1$. If $u_i(\x) \le \eps/n$, then $a^\eps_{ji}(\x) = 1$. Otherwise, we have $u_i(\x - a^\eps_{ji}(\x)\cdot \e_{ji}) = u_i(\x) - \eps/n$. By the Lipschitz continuity of $u_i$, we get $a^{\eps}_{ji}(\x) \ge \eps/nL$. In either case we see that $a^{\eps}_{ji}(\x) \ge \eps/nL$ for all $i,j\in \N$ and $\x\in I_{ji}$. Since $\alpha(\eps) = \min_{i,j\in \N} \inf_{\x\in I_{ji}} a^\eps_{ji}$, we obtain $\alpha(\eps) \ge \frac{\eps}{nL}$, proving part (ii).

For part (iii), consider an exchange $\x$ such that $G(\x, \frac{\eps}{nL})$ is acyclic. Then for any cycle $C$ in $K_n$, there is some edge $e\in C$ such that $x_e \ge 1 - \eps/nL$. Since $\alpha(\eps) \ge \eps/nL$ from part (ii), we also get $x_e \ge 1-\alpha(\eps)$. Thus the exchange graph $G(\x, \alpha(\eps))$ is acyclic. By \cref{lem:acyclic-cs}, $\x$ is an $\eps$-\cs exchange.
\end{proof}

Moreover, we assume oracle access to the utility functions and shares, i.e., given a rational polynomially bounded $x$, we get oracle access to $u_i(x_i)$ and $\psi_{ij}(x)$. We remark that all our algorithmic guarantees hold, even when we have access to approximate utilities and shares, where approximations are {polynomial in $\varepsilon, 1/L$, and $1/n$}. However, we stick to oracle access to exact utilities and shares for clarity and ease of presentation.

We first present a local search algorithm, which will later be used to show the membership in PLS. Then, we adapt our fixed point formulation in \cref{sec:existence}, to show membership in PPAD, implying that the problem lies in PPAD $\cap$ PLS $=$ CLS.

\paragraph{Perturbation.} 
To prove the results presented in this section, we perform a perturbation of the utilities as follows: for each $i\in \N$, we define $\tilde{u}_i(\x_i) = u_i(\x_i) + (\varepsilon \cdot \sum_j x_{ji})/n$. Also, we introduce a perturbed function $\tilde{\psi}_{ij}(\x) = \psi_{ij}(\x) + \varepsilon x_{ij}/n$. It can be verified that $\tilde{\psi}_{ij}(x)$ satisfies monotonicity, normalization and efficiency and is, therefore, a valid share function. Further, as long as $\psi_{ij}(x)$ is cross-monotone, so is $\tilde{\psi}_{ij}(x)$. Therefore, with $\tilde{u}(\cdot)$ and $\tilde{\psi}_{ij}(\cdot)$, we now have a perturbed instance with monotone utility functions, and cross monotone share functions. 

The primary reason behind the perturbation is to incorporate mild non-satiation for the utilities, i.e., increasing data consumption should cause a strict increase in the utility value at any point. We first claim that any $\varepsilon$-core-stable and $\varepsilon$-reciprocal exchange $x$ w.r.t. to the perturbed utilities will give a $\mathcal{O}(\varepsilon)$-core stable and $\mathcal{O}(\varepsilon)$-reciprocal exchange w.r.t. the original utilities.  

\begin{lemma}
    \label{lem:perturbation}
    If $x$ is $\varepsilon$-core-stable and $\varepsilon$-reciprocal exchange $x$ w.r.t. to the perturbed utilities, then $x$ is a $3\varepsilon$-core stable and $3\varepsilon$-reciprocal exchange w.r.t. the original utilities.  
\end{lemma}

\begin{proof}
    First observe that for any $\x$, we have $|\tilde{u}_i(\x) - {u}_i(\x)| \leq {\varepsilon \sum_{j} x_{ji}}/{n}$ for all $i$.  Also,  we have $|\tilde{\psi}_{ij}(\x) - \psi_{ij}(\x)| = \varepsilon x_{ij}/n$. Finally, let $\tilde{\Delta}_i(\x)$ denote the surplus of agent $i$ w.r.t. $\tilde{u}_i(\cdot)$. Observe that
    \begin{align*}
     |\tilde{\Delta}_i(\x) - \Delta_i(\x)| &\leq |\tilde{u}_i(\x) - {u}_i(\x)| + \sum_j|\tilde{\psi}_{ij}(\x) - \psi_{ij}(\x)|\\
     &\leq \frac{\varepsilon \sum_j x_{ji}}{n} +  \frac{\varepsilon \sum_j x_{ij}}{n}\\
     &\leq {2\varepsilon}. 
    \end{align*}
   Thus, if $\x$ is $\varepsilon$-reciprocal w.r.t. perturbed utilities (means that $|\tilde{\Delta}_i(\x)| \leq \varepsilon$ for all $i$), then $\x$ is $\varepsilon + 2\varepsilon = 3\varepsilon$ reciprocal w.r.t. original utilities ($|\Delta_i(\x)| \leq |\tilde{\Delta}_i(\x)| + 2\varepsilon \leq 3\varepsilon$). Similarly, if $\x$ is $\varepsilon$-core stable w.r.t. perturbed utilities, then $\x$ is $\varepsilon + 2\varepsilon = 3\varepsilon$ core-stable w.r.t. original utilities: since $\x$ is core-stable, there exists no $S \subseteq [n]$ and no exchange $\y$ among agents in $S$ such that $\tilde{u}_i(\y) > \tilde{u}_i(\x) + \varepsilon$. Since $|\tilde{u}_i(\z) -{u}_i(\z)| \leq \varepsilon \sum_j z_{ji}/n \leq \varepsilon$ for $\z \in \{\x,\y\}$, we have that there exists no $S \subseteq [n]$ and no $\y$ among agents in $S$, such that $u_i(\y) > u_i(\x) + 3\varepsilon$.
\end{proof}

Therefore, in what follows, we provide a local search algorithm and prove PLS membership only for perturbed instances. In particular, in all instances we consider from here on, increasing any $x_{ij}$ by $\delta$ increases $u_j(\x_j)$ by at least $\varepsilon \delta /n $.

\subsection{Local Search Algorithm}
Throughout the algorithm, we maintain the invariant that $G(\x, \varepsilon/(nL))$ is acyclic. This ensures that our $\x$ is always $\varepsilon$-core stable by \cref{lem:stronger-assumptions}. We will start out the algorithm by setting $x_{ij} = 1$ for all $i,j$, in which case $G(\x,\varepsilon/(nL))$ is trivially acyclic. 

To achieve reciprocity, our algorithm will reduce the high surpluses in an exchange, while still maintaining $G(\x,\varepsilon/(nL))$ is acyclic. Given $\x$, let $\Delta(\x) = \langle \Delta_{\sigma(1)}(\x), \Delta_{\sigma(2)}(\x), \dots, \Delta_{\sigma(n)}(\x) \rangle$ be the \emph{sorted surplus profile} of $\x$, where $\sigma$ is a permutation of $[n]$ s.t. $\Delta_{\sigma(1)}(\x) \geq \Delta_{\sigma(2)}(\x) \geq \dots \geq \Delta_{\sigma(n)}(\x)$. Define $P(\x) = \sum_{i \in \N} (nL/\varepsilon)^{2(n-i)} \cdot \Delta_{\sigma(i)}(\x)$. Given an exchange $\x$ which is not $\varepsilon$-reciprocal, our algorithm determines an exchange $\x'$ such that $G(\x', \varepsilon/(nL))$ is acyclic, and $P(\x') < P(\x) - L/2n^2$. 

Our algorithm first identifies a set $S$ of agents with high surplus as follows: let $i$ be an agent with maximum surplus. If $\Delta_i(\x) \leq \varepsilon/n$, terminate and return $\x$ as the desired exchange. Otherwise, let $S \gets \{i\}$. Let $i'$ be the agent in  $\N \setminus S$ with maximum surplus. While $\Delta_{i'}(\x) \geq \min_{i \in S} \Delta_i(\x) - \varepsilon/n^2$, add $i'$ to $S$. Note that at termination, $S$ will only comprise of agents with strictly positive surplus, as $\min_{i \in S} \Delta_i(\x) \geq \max_{i \in S} \Delta_i(\x) - |S| \cdot \varepsilon/n^2 > \varepsilon/n - n \cdot \varepsilon/n^2 = 0$. This also implies that $|S| < n$ as if there is an agent with strictly positive surplus, there must be at least one agent with strictly negative surplus. This simple claim would be useful later in our analysis.

\begin{algorithm}[t]
\caption{Selecting $S$}\label{alg-S}
\begin{algorithmic}[1]
\Statex \textbf{Input:} Exchange $\x$
\Statex \textbf{Output:} Set $S$
\State $i \gets \arg\max_j \Delta_j(\x)$
\State $S \gets \{i\}$
\While{$\max_{j \in \N \setminus S} \Delta_j(\x) \geq \min_{i \in S} \Delta_i(\x) - \varepsilon/n^2$ }
\State $j \gets \argmax_{j \in \N\setminus S} \Delta_j(\x)$.
\State $S \gets S \cup \{j\}$
\EndWhile
\State \Return $S$
\end{algorithmic}
\end{algorithm}

\begin{claim}
    \label{claim:S&n}
     We have $|S| < n$.
\end{claim}

Intuitively, set $S$ is the set of agents with high surplus. In every iteration, our algorithm will determine an allocation $\x'$ such that 
\begin{enumerate}[label=(\roman*)]
    \item $\Delta_i(\x') \leq \Delta_i(\x)$ for all $i \in S$ with at least one agent $i^* \in S$ having $\Delta_i(\x') \leq \Delta_i(\x) - \gamma$, for some sufficiently positive $\gamma$ and 
    \item the surplus of no agent in $\N \setminus S$ will be more than the surplus of any agent in $S$, i.e.,  $\max_{i \in \N \setminus S} \Delta_i(\x') \leq \min_{i \in S}\Delta_i(\x')$. 
\end{enumerate}

We show that properties (i) and (ii) with an appropriate choice of $\gamma$, ensure $P(\x') \leq P(\x)- L/2n^2$. We determine an $\x'$ satisfying (i) and (ii), by either reducing the flow from agents in $S$ to $\N \setminus S$, or increasing the flow from agents in $\N \setminus S$ to $S$, while still maintaining acyclicity of the exchange graph $G(\x, \varepsilon/(nL))$. 

\subsubsection{Decreasing Data Flow from \texorpdfstring{$S$ to $\N \setminus S$}{S to N minus S}} Note that decreasing flow from any agent $i \in S$ to any agent $j \in \N \setminus S$ may create cycles in $G(\x, \varepsilon/(nL))$, as there may be a path from $j$ to $i$ in $G(\x, \varepsilon/(nL))$. \emph{Therefore, we only decrease data flow from $S$ to $\N \setminus S$ when there is no edge from $\N \setminus S$ to $S$ in $G(\x, \varepsilon/(nL))$.} 

We first observe that there exists an agent $j \in \N \setminus S$, who gets sufficient utility-flow from agents in $S$. 

\begin{lemma}\label{obs:non-zeroflow}
If $\max_i \Delta_i(\x) > \eps/n$, there exists an agent $j \in \N \setminus S$, such that $\sum_{i \in S} \psi_{ij}(\x) > \varepsilon/n^2$.
\end{lemma}

\begin{proof}
Clearly, $S$ comprises of the agent with highest surplus, say $i^*$. We have $\sum_{i \in S} \Delta_i(\x) \geq \Delta_{i^*}(\x) > \varepsilon/n$. Since $\sum_{i \in \N} \Delta_i(\x)  =  \sum_{i \in S} \Delta_i(\x) + \sum_{i \in \N \setminus S} \Delta_i(\x) = 0$, and $\sum_{i \in S} \Delta_i(\x) > \varepsilon/n$, we have $\sum_{i \in \N \setminus S} \Delta_i(\x) < -\varepsilon/n$. Also note that:
\begin{align*}
   \sum_{i \in \N \setminus S}\Delta_i(\x) &= \sum_{i \in \N \setminus S} \sum_{j \in \N} \psi_{ij}(\x) - \sum_{ i \in \N \setminus S} u_i(\x)\\
    &=\bigg(\sum_{i \in \N \setminus S} \sum_{j \in \N \setminus S} \psi_{ij}(\x) +  \sum_{i \in \N \setminus S} \sum_{j \in S} \psi_{ij}(\x) \bigg) - \bigg(\sum_{ i \in \N \setminus S} \sum_{j \in \N\setminus S} \psi_{ji}(\x) + \sum_{ i \in \N \setminus S} \sum_{j \in S} \psi_{ji}(\x) \bigg)\\
    &=\sum_{i \in \N \setminus S} \sum_{j \in S} \psi_{ij}(\x) - \sum_{i \in S} \sum_{j \in \N \setminus S} \psi_{ij}(\x)\\
    &> -\sum_{i \in S} \sum_{j \in \N \setminus S} \psi_{ij}(\x)
\end{align*}
Therefore, we have $\sum_{i \in S} \sum_{j \in \N \setminus S} \psi_{ij}(\x) > \varepsilon/n$. This implies that there exists a $j \in \N \setminus S$ such that $\sum_{i \in S} \psi_{ij}(\x) > \varepsilon/n^2$. 
\end{proof}

The next step in our algorithm is to reduce the flow from $S$ to $j\in \N\setminus S$. Note that due to cross-monotonicity, if we reduce $x_{ij}$ for some $i \in S$, this could cause $\psi_{i'j} (\x)$ (and consequently $\Delta_{i'}(\x)$) for some $i' \in S$ and $i' \neq i$ to increase. To combat this problem, we will not decrease any $x_{ij}$ such that it can cause a huge reduction in the surplus of agent $i$. In particular, we do not want to reduce any $\Delta_i(\x)$ more that  $\varepsilon/2n^3$, in an iteration. So for each agent $i$, we set our lower bound $\delta_i = \Delta_i - \varepsilon/2n^3$. Our algorithm is as follows: while there exists an agent $i$ such that $\Delta_i(\x) - \delta_i \geq \varepsilon/4n^3$ and $x_{ij} > 0$, then reduce $x_{ij}$ until either $x_{ij} = 0$ or $\Delta_i(\x) = \delta_i$\footnote{This is achieved through binary search. Through binary search, we can determine $x_{ij}$ such that $\Delta_i(\x) =\delta_i + \mathit{err}_{ij}$, where $\mathit{err}_{ij} =(\varepsilon/nL)^{\textup{poly}(n)}$. We assume that $\mathit{err}_{ij} = 0$ for ease of presentation. One can verify that all our guarantees (Lemmas~\ref{lem:termination-case1},~\ref{lem:no-increase-surplus},~\ref{lem:decrease-surplus},~\ref{lem:no-cross},~\ref{lem:main-technical},~\ref{lem:main-case-2}) work with $\mathit{err}_{ij} = (\varepsilon/nL)^{\textup{poly}(n)}$.}. Note that we only reduce the data flow throughout this phase from $S$ to $\N \setminus S$. We first show that this procedure terminates in time polynomial in $n, 1/\varepsilon, L$.

\begin{algorithm}[t]
\caption{Decreasing flow from $S$ to $\N \setminus S$}\label{alg-D}
\begin{algorithmic}[1]
\Statex \textbf{Input:} Exchange $\x$
\Statex \textbf{Output:} Exchange $\x'$
\State \textbf{Determine} $S$, by Algorithm~\ref{alg-S}($\x$).
\If{$S = [n]$}
\State \Return $\x$.
\Else 
\State Set $\delta_i = \Delta_i(\x) - \varepsilon/2n^3$ for all $i \in S$.
\While{$\exists i \in S$ s.t. $\Delta_i(\x) - \delta_i > \varepsilon/4n^3$}
\State $x_{\delta} \gets$ smallest $r$ such that $\Delta_i(\x - r \cdot \mathbf{e}_{ij}) = 0$. {If no such $r$ exists, then set $x_{\delta} = x_{ij}$}  
\State $\x = \x - x_{\delta}\cdot\mathbf{e}_{ij}$ 
\State \textbf{Update} $\Delta_i(\x)$  for all $i \in [n]$.
\EndWhile
\State \Return $\x$
\EndIf
\end{algorithmic}
\end{algorithm}

\begin{lemma}
    \label{lem:termination-case1}
      Algorithm~\ref{alg-D} terminates in $\textup{poly}(n, 1/\varepsilon, L)$ time.
\end{lemma}

\begin{proof}
In every iteration of Algorithm~\ref{alg-D}, some $x_{ij} = 0$ or some $\Delta_i(\x) = \delta_i$. The former ($x_{ij} = 0$) can happen only $n$ times. For each time we reduce the surplus of an agent $i$ to $\delta_i$, we decrease $\Delta_i(\x)$ by at least $\varepsilon/4n^3$ (otherwise, we would have terminated our phase). Equivalently, we would have decreased $\psi_{ij}(\x)$ by at least $\varepsilon/4n^3$ by decreasing $x_{ij}$. Since the gradients of the utility functions are upper-bounded by $L$, this implies that $x_{ij}$ is reduced by at least $\varepsilon/(4n^3L)$. Thus, every time $\Delta_i(\x) = \delta_i$, $x_{ij}$ reduces by at least $\varepsilon/4n^3L$, implying that the number of such iterations is $\textup{poly}(n,1/\varepsilon,L)$.
\end{proof}

Note that in an iteration, if $\Delta_i(\x)$ is reduced to $\delta_i$, it can increase in the future iterations of Algorithm~\ref{alg-D}, as some $x_{i'j}$ with $i' \neq i$ may decrease, causing an increase in $\psi_{ij}(\x)$ (and consequently $\Delta_i(\x)$). However, we prove that no surpluses in $S$ would have increased at termination.

\begin{lemma}
    \label{lem:no-increase-surplus}
     When Algorithm~\ref{alg-D} terminates, no surplus in $S$ will increase, i.e., $\Delta_i(\x') \leq \Delta_i(\x)$ for all $i \in S$.
\end{lemma}

\begin{proof}
 Assume otherwise. There exists an agent $i \in S$ such that $\Delta_i(\x') > \Delta_i(\x)$. Then, we have $\Delta_i(\x') - \delta_i > \Delta_i(\x) - \delta_i = \varepsilon/2n^3$. Thus, the only reason, the algorithm did not decrease flow from $i$ to $j$ is $x'_{ij} = 0$. However if $x'_{ij} = 0$, then $\psi_{ij}(\x') \leq \psi_{ij} (\x)$. Since $\psi_{ij'}(\x') = \psi_{ij'}(\x)$ for all other $j' \neq j$ (as we do not alter the data flow incoming to agent $j'$), we have
 \[
     \Delta_i(\x') = \sum_{j} \psi_{ij}(\x') - u_i(\x'_i)\leq \sum_{j} \psi_{ij}(\x) - u_i(\x_i)=\Delta_i(\x),
\]
 which is a contradiction.
\end{proof}

Observe that Algorithm~\ref{alg-D} has an inverse effect on the surpluses of the agents in $\N \setminus S$. Throughout the algorithm, for any agent $i \in \N \setminus S$, $\sum_j \psi_{ij}(\x)$ can only increase (when data inflow to some other agent in $\N \setminus S$ is decreased), and $u_i(\x_i)$ can only decrease (if data inflow to $i$ from $S$ is decreased), implying that the surpluses of the agents in $\N \setminus S$ cannot decrease throughout Algorithm~\ref{alg-D}.   

\begin{observation}\label{obs:low-surplusincrease}
When Algorithm~\ref{alg-D} terminates, no surpluses in $\N \setminus S$ can decrease, i.e.,  we have $\Delta_i(\x') \geq \Delta_i(\x)$ for all $i \in \N \setminus S$.
\end{observation}

Now, we show that at least one agent in $S$ will witness a strict reduction in its surplus.

\begin{lemma}\label{lem:decrease-surplus}
When Algorithm~\ref{alg-D} terminates, at least one surplus in $S$ will strictly decrease by $\varepsilon/4n^3$.
\end{lemma}
\begin{proof}
At the beginning of the procedure, $\sum_{i \in S} \psi_{ij} (\x) > \varepsilon/n^2$, implying that there exists a $i \in S$ such  that $\psi_{ij}(\x) > \varepsilon/n^3$. At the end of the procedure, either $\Delta_i(\x) - \delta_i < \varepsilon/4n^3$ or $x_{ij} = 0$. If it is the former, then the surplus has reduced by at least $\varepsilon/4n^3$ (as initially there was a gap of $\varepsilon/2n^3$), and if it is the latter, then $\Delta_i(\x') - \Delta_i(\x) = \psi_{ij}(\x) > \varepsilon/n^3 \geq \varepsilon/4n^3$. 
\end{proof}

\noindent In order to reduce our potential function $P(\x)$, we need to reduce the high surpluses. It is therefore important to argue that while we are decreasing surpluses of agents in $S$, we do not create any new agents with very high surplus outside of $S$. In the following lemma, we prove that this is exactly the case: we ensure that all surpluses in $\N \setminus S$ are strictly lower than any surplus in $S$ at the termination of Algorithm~\ref{alg-D}.

\begin{lemma}
\label{lem:no-cross}
When Algorithm~\ref{alg-D} terminates, all agents in $\N \setminus S$ will have strictly lower surplus than any agent in $S$, i.e., we have $\max_{i \in \N \setminus S} \Delta_i(\x) < \min_{i \in S} \Delta_i(\x)$.
\end{lemma}
\begin{proof}
At the beginning of the algorithm, we have $\min_{i \in S} \Delta_i(\x) > \max_{i \in \N \setminus S} \Delta_i(\x) + \varepsilon/n^2$. Since we have $\sum_i \Delta_i(\x') = \sum_i \Delta_i(\x) = 0$, we have the total increase in the surpluses of agents in $\N \setminus S$ (when moving from exchange $\x$ to $\x'$) is equal to the total decrease in surpluses of the agents in $S$. Since no surplus in $S$ decreases more than $\varepsilon/2n^3$, the total decrease in surpluses in $S$ is at most $|S| \cdot \varepsilon/2n^3 \leq \varepsilon/2n^2$. This implies that the surplus of any agent in $S$ can decrease at most by $\varepsilon/2n^2$, and the total surplus of any agent in $\N \setminus S$ can increase by at most $\varepsilon/2n^2$. Since the initial gap was strictly more than $\varepsilon/n^2$, all agents in $\N \setminus S$ will have a strictly lower surplus than any agent in $S$ at the end of the algorithm.
\end{proof}

We are now ready to show that our Algorithm~\ref{alg-D} improves $P(\x)$ by at least $L/2n^2$. Before we present the proof, we introduce a technical lemma, which will also help us when we discuss the section on increasing flow from $\N \setminus S$ to $S$.

\begin{lemma}
  \label{lem:main-technical}
   Given two exchanges $\x$ and $\x'$, and a set $S \subseteq \N$ such that
    \begin{itemize}
        \item $\min_{i \in S}\Delta_i(\z) \geq \max_{i \in \N \setminus S} \Delta_i(\z)$ for $\z \in \{\x, \x'\}$,
        
        \item $\Delta_i(\x) \geq \Delta_i(\x')$ for all $i \in S$.

        \item Let $\sum_{i \in S} (\Delta_i(\x) - \Delta_i(\x')) = h_1$, and $\sum_{i \in \N \setminus S} \max(0, \Delta_i(\x') - \Delta_i(\x)) = h_2$, and $h_1 \geq 2\frac{h_2}{(nL/\varepsilon)^2}$.
 
    \end{itemize}
   Then $P(\x') \leq P(\x) - \frac{h_1}{2} \cdot (\frac{nL}{\varepsilon})^2$.
\end{lemma}
\begin{proof}
Let $\Delta(\x) = \langle \Delta_{\sigma(1)}(\x), \Delta_{\sigma(2)}(\x), \dots, \Delta_{\sigma(n)}(\x) \rangle$ and $\Delta(\x') = \langle \Delta_{\sigma'(1)}(\x'), \Delta_{\sigma'(2)}(\x'),  \dots,~~$ $\Delta_{\sigma'(n)}(\x') \rangle$ be the sorted surplus profiles of exchanges $\x$ and $\x'$ respectively. Since $\min_{i \in S}\Delta_i(\z) \geq \max_{i \in [n] \setminus S} \Delta_i(\z)$ for $\z \in \{\x,\x'\}$, this implies that $\sigma(i) \in S$ and $\sigma'(i) \in S$ for all $i \in [|S|]$. Symmetrically,  $\sigma(i)  \in \N \setminus S$ and $\sigma'(i) \in \N \setminus S$ for all $i \in \N \setminus [|S|]$. We now make a technical claim.

\begin{claim}\label{claim:technical}
We have $\Delta_{\sigma'(i)}(\x) \geq \Delta_{\sigma(i)}(\x')$ for all $i \in [|S|]$.
\end{claim}
\begin{proof}
Assume otherwise. Say there exists a $ i \in [|S|]$ such that $\Delta_{\sigma'(i)}(\x') > \Delta_{\sigma(i)}(\x)$. This implies that there exist $i$ agents, namely $\sigma'(1), \dots, \sigma'(i)$, in $\x'$ that have surpluses strictly larger than $\Delta_{\sigma(i)}(\x)$. 
Note that all these $i$ agents belong to $S$, as $\sigma'(i) \in S$ for all $i \in [|S|]$. Furthermore, we have $\Delta_i(\x) \geq \Delta_i(\x')$ for all $i \in S$, implying that even in $\x$, there are at least $i$ agents having surplus strictly larger than $\Delta_{\sigma(i)}(\x)$, which is a contradiction as by definition of $\sigma(\cdot)$, there are at most $(i-1)$ surpluses larger than $\Delta_{\sigma(i)}(\x)$ in $\x$. 
\end{proof}

From Claim~\ref{claim:technical}, we immediately have
\[
\sum_{i \in [|S|]}(nL/\varepsilon)^{2(n-i)} \Delta_{\sigma'(i)}(\x') \leq \sum_{i \in [|S|]}(nL/\varepsilon)^{2(n-i)} \Delta_{\sigma(i)}(\x) .
\] Recall that $h_1$ denotes the total decrease in the surpluses of agents in $S$ when we change the exchange from $\x$ to $\x'$, i.e., $h_1 = \sum_{i \in S} (\Delta_i(\x) - \Delta_i(\x'))$. Then, note that we have,
\begin{align}
\label{eq-1}
\sum_{i \in [|S|]}(nL/\varepsilon)^{2(n-i)} \Delta_{\sigma'(i)}(\x') \leq \sum_{i \in [|S|]}(nL/\varepsilon)^{2(n-i)} \Delta_{\sigma(i)}(\x) - (nL/\varepsilon)^{2(n-|S|)}h_1    
\end{align}

Now, recall that $h_2$ denotes the sum of increases of all surpluses in $[n] \setminus S$ that have increased from $\x$ to $\x'$, i.e.,  $h_2 = \sum_{i \in [n] \setminus S} \max(0, \Delta_i(\x') - \Delta_i(\x))$
Therefore, we have  
\begin{align}
 \label{eq-2}
\sum_{i \in [n] \setminus [|S|]}(nL/\varepsilon)^{2(n-i)} \Delta_{\sigma'(i)}(\x') \leq \sum_{i \in [n] \setminus [|S|]}(nL/\varepsilon)^{2(n-i)} \Delta_{\sigma(i)}(\x) + (nL/\varepsilon)^{2(n-|S|-1)}h_2    
\end{align}
Equations~(\ref{eq-1}) and~(\ref{eq-2}) imply that

\begin{align*}
 \sum_{i \in \N}(\frac{nL}{\varepsilon})^{2(n-i)} \Delta_{\sigma'(i)}(\x') &\leq \sum_{i \in \N}(\frac{nL}{\varepsilon})^{2(n-i)} \Delta_{\sigma(i)}(\x) - h_1(\frac{nL}{\varepsilon})^{2(n-|S|)} + h_2 (\frac{nL}{\varepsilon})^{2(n-|S|-1)}\\
 &\leq \sum_{i \in \N}(\frac{nL}{\varepsilon})^{2(n-i)} \Delta_{\sigma(i)}(\x) - (h_1 -\frac{h_2}{(nL/\varepsilon)^2}) (\frac{nL}{\varepsilon})^{2(n-|S|)} \\
 &\leq \sum_{i \in \N}(\frac{nL}{\varepsilon})^{2(n-i)} \Delta_{\sigma(i)}(\x) - \frac{h_1}{2} (\frac{nL}{\varepsilon})^{2(n-|S|)} \\ 
 &< \sum_{i \in \N}(\frac{nL}{\varepsilon})^{2(n-i)} \Delta_{\sigma(i)}(\x) - \frac{h_1}{2} (\frac{nL}{\varepsilon})^{2}
\end{align*}   
where the last inequality follows from the fact that $|S| < n$ (by Claim~\ref{claim:S&n}).
\end{proof}
 
We now show that Lemma~\ref{lem:main-technical} will immediately imply that $\x'$ returned by our algorithm satisfies $P(\x') < P(\x) - L^2/(8n\varepsilon) < P(\x) - L/(2n^2)$. Note that $\x'$ returned by the algorithm satisfies the pre-conditions mentioned in Lemma~\ref{lem:main-technical} by setting $h_1 = h_2$: Lemma~\ref{lem:no-increase-surplus} ensures $\Delta_i(\x') \leq \Delta_i(\x)$ for all $i \in S$, Lemma~\ref{lem:no-cross} ensures that $\max_{i \in \N \setminus S} \Delta_i(\x) < \min_{i \in S} \Delta_i(\x)$. Finally, Observation~\ref{obs:low-surplusincrease} ensures that for each $i \in \N \setminus S$, we have $\Delta_i(\x') \geq \Delta_i(\x)$, implying that $h_2 = \sum_{i \in \N \setminus S} (\Delta_i(\x') - \Delta_i(\x))$. Therefore, 
 \begin{align*}
     h_2 &= \sum_{i \in \N \setminus S} (\Delta_i(\x') - \Delta_i(\x))\\
     &=\sum_{i \in S} (\Delta_i(\x) - \Delta_i(\x')) &\text{(since $\sum_i \Delta_i(\x) = \sum_i \Delta_i(\x')=0$)}\\
     &=h_1
 \end{align*}
When $h_1= h_2$, $h_1 \geq \frac{h_2}{(nL/\varepsilon)^2}$ trivially holds. Therefore, Lemma~\ref{lem:main-technical} immediately implies that $P(\x') < P(\x) - \frac{h_1}{2}(\frac{nL}{\varepsilon})^2$. Lemma~\ref{lem:decrease-surplus} ensures that $h_1 \geq \varepsilon/4n^3$. Substituting $h_1$, we immediately arrive at the following corollary.  

\begin{corollary}\label{cor:maincase-1}
Let $\x'$ be the allocation returned by Algorithm~\ref{alg-D}. We have $P(\x') < P(\x) - \frac{L^2}{8n\varepsilon} < P(\x) - \frac{L}{2n^2}$.      
\end{corollary}

\subsubsection{Increasing Data Flow from \texorpdfstring{$\N \setminus S$ to $S$}{N minus S to S}} Recall that we only decrease flow  from $S$ to $\N \setminus S$ when there are no edges in $G(\x, \varepsilon/(nL))$ from $\N \setminus S$ to $S$. Now, we consider the case when there are edges from $\N \setminus S$ to $S$ in $G(\x, \varepsilon/(nL))$. 

We consider the case when there exists an agent $j \in [n] \setminus S$ and there is an edge $\overrightarrow{(j,i)} \in G(\x, \varepsilon/(nL))$. This implies that $x_{ji} < 1 - \varepsilon/nL$, i.e.,  there is room in increasing $x_{ji}$. Note that increasing $x_{ji}$ does not violate the acyclicity of $G(\x, \varepsilon/(nL))$, as it does not introduce any new edges. Furthermore, note that increasing $x_{ji}$ causes $\Delta_j(\x)$ to increase, and $\Delta_{j'}(\x)$ for all $j' \neq j$ to decrease. In particular, this means that the surpluses of all agents in $S$ can only decrease when increasing $x_{ij}$. Thus, a careful increase may help us reduce $P(\x)$. We now consolidate this intuition.

To this end, our algorithm increases $x_{ji}$ by $\varepsilon/n^3L$: note that this operation is feasible as $x_{ji} < 1 - \varepsilon/nL$. The algorithm is outlined in Algorithm~\ref{alg-I}. Let $\x'$ be the exchange returned by Algorithm~\ref{alg-I}.

\begin{algorithm}[t]
\caption{Increasing Flow from $\N \setminus S$ to $S$}\label{alg-I}
\begin{algorithmic}[1]
\Statex \textbf{Input:} Exchange $\x$
\Statex \textbf{Output:} Exchange $\x'$
\State Pick $j \in [n] \setminus S$ such that $\overrightarrow{(j,i)} \in G(x, \varepsilon/nL)$ where $i \in S$.
\State $x_{ji} \gets x_{ji} + \varepsilon/n^3L$.
\State \Return $x$
\end{algorithmic}
\end{algorithm}

\begin{lemma}\label{lem:main-case-2}
Let $\x'$ be the exchange returned by Algorithm~\ref{alg-I}. Then we have $P(\x') < P(\x) - L/2n^2$.
\end{lemma}

\begin{proof}
First observe that increasing $x_{ji}$ can only reduce $\Delta_{\ell}(\x)$ for all $\ell \in S$ as $\psi_{\ell i}(\x)$ can only reduce for all $\ell \neq j$, and $\psi_{\ell i'}(\x)$ remains unchanged for all $i' \neq i$. Therefore, we have in $\Delta_{\ell}(\x') \leq \Delta_{\ell}(\x)$ for all $\ell \in S$.

Further, since $x'_{ji} = x_{ji} + \varepsilon/n^3L$, we have $\psi_{ji}(\x') \leq \psi_{ji}(\x) + L \cdot \varepsilon/n^3L = \psi_{ji}(\x) + \varepsilon/n^3$, implying $\Delta_j(\x') \leq \Delta_j(\x) + \varepsilon/n^3$. Note that since $\sum_{\ell} \Delta_{\ell}(\x') = \sum_{\ell} \Delta_{\ell}(\x) = 0$, and $\Delta_j(\x)$ is the only surplus that increases, we have the total decrease in the surpluses of agents in $S$ is at most the total increase in surplus of $j$ which is $\varepsilon/n^3$. This implies that the surplus of any agent in $S$ can decrease at most by a factor of $\varepsilon/n^3$. Since $\min_{\ell \in S}\Delta_{\ell}(\x) \geq \max_{\ell \in \N\setminus S} \Delta_{\ell}(\x) + \varepsilon/n^2$, we will have $\min_{\ell \in S}\Delta_{\ell}(\x') > \max_{\ell \in \N \setminus S} \Delta_{\ell}(\x')$.

Finally let $h_1 = \sum_{\ell \in S} (\Delta_{\ell}(\x)- \Delta_{\ell}(\x'))$, and $h_2 = \sum_{\ell \in \N \setminus S} \max(0, \Delta_{\ell}(\x') - \Delta_{\ell}(\x))$. Since the only surplus that increases from $\x$ to $\x'$ in $\N \setminus S$ is that of agent $j$, we have $h_2 = \Delta_j(\x') - \Delta_j(\x) \leq \varepsilon/n^3$. Since we are dealing with perturbed utilities, we have $h_1 \geq \Delta_i(\x) -\Delta_i(\x') \geq (\varepsilon/n) \cdot (\varepsilon/n^3L) = \varepsilon^2/(n^4L)$. Note that we still have $h_1 \geq \frac{2h_2}{(nL/\varepsilon)^2}$. Applying Lemma~\ref{lem:main-technical}, we have that at the end of Algorithm~\ref{alg-I}, $P(\x') < P(\x) - \frac{h_1}{2} \cdot (\frac{nL}{\varepsilon})^2$. Substituting $h_1 \geq \varepsilon^2/n^4L$, we obtain:
\[ P(\x') < P(\x) - \frac{\varepsilon^2}{2n^4L} \cdot (\frac{nL}{\varepsilon})^2 = P(\x) - \frac{L}{2n^2} \qedhere\]
\end{proof}

\paragraph{Summary.} We now outline our entire local search algorithm: we initialize with $x_{ij} = 1$ for all $i,j$. Then $G(\x, \varepsilon/nL)$ is trivially acyclic. Then iteratively, our algorithm (i) identifies the set $S$ by Algorithm~\ref{alg-S}, (ii) determines $\x'$ from $\x$ by Algorithm~\ref{alg-I} if there are any edges in $G(\x, \varepsilon/nL)$ from $[n] \setminus S$ to $S$, and (iii) determines $\x'$ from $\x$ by Algorithm~\ref{alg-D} otherwise. Every iteration, our algorithm improves $P(\x)$ by at least $L/2n^2$. The entire procedure is outlined in Algorithm~\ref{alg-LS}.

\begin{algorithm}[t]
    \caption{Local Search Algorithm}\label{alg-LS}
    \begin{algorithmic}[1]
        \State Set $x_{ij} = 1$ for all $i,j \in \N$.
        \State Compute $\Delta_i(\x)$ for all $i \in \N$
         \While{$\max_{i} \Delta_i(\x) > \varepsilon/n$}
            \State $S \gets $ Algorithm~\ref{alg-S}($\x$)
            \If{$\exists j \in \N\setminus S$ and $i \in S$ such that $\overrightarrow{(j,i)} \in G(\x, \varepsilon/nL)$}
               \State $x \gets $ Algorithm~\ref{alg-I}($\x$).
            \Else
              \State $x \gets $ Algorithm~\ref{alg-D}($\x$).
             \EndIf 
          \EndWhile 
        \State \Return $\x$
    \end{algorithmic}
\end{algorithm}

\begin{theorem}\label{thm:localsearch}
Given an exchange $\x$ which is not $\varepsilon$-core stable and $\varepsilon$-reciprocal, Algorithm~\ref{alg-LS} constructs an allocation $\x'$  such that $P(\x') < P(\x) - L/(2n^2)$ in time $\textup{poly}(n,1/\varepsilon, L)$.    
\end{theorem}

\subsection{Membership in PLS}

 We consider a discrete space of exchanges, where we restrict every $x_{ij}$ to be an integer multiple of $(\varepsilon/nL)^{n^2+1}$.  In particular let $X$ denote the finite set of all exchanges where $x_{ij}$ is an integer multiple of $(\varepsilon/nL)^{n^2+1}$, and $G(\x, \varepsilon/(nL))$ is acyclic. We assume without loss of generality that $(nL/\varepsilon)^{n^2+1}$ is an integer, implying that $1$ is an integer multiple of $(\varepsilon/nL)^{(n^2+1)}$.

For every $\x \in X$, the associated cost function is $P(\x)$, which can be computed in time $\textup{poly}(n,\log(L))$. Note that an exchange $x$ with $x_{ij}=1$ for all $i,j$ is in $X$, so we can identify an element in $X$ in time $\textup{poly}(n)$.

Finally,  for every $\x \in X$ which is not acyclic, we show how to output an $\x' \in X$  in time $\textup{poly}(n,L,1/\varepsilon)$ such that $P(\x') < P(\x)$. This will show that determining a $\varepsilon$-core stable and $\varepsilon$-reciprocal exchange is in PLS when $L/\varepsilon = \textup{poly}(n)$. Given every feasible $\x \in X$, we can define the neighborhood of $\x$ in $X$ as the solution returned by the foregoing algorithm. This way, whenever $\x$ has a neighbor $\x'$, the algorithm can pick $\x'$ and we have $P(\x') < P(\x)$. Whenever $\x$ has no neighbor, then we have a local optimum. We now show the algorithm with the desired properties.

For this, consider any $\x$ in $X$. Let $\tilde{\x}$ be the output of Algorithm~\ref{alg-LS} starting with $\x$. We get $\x'$ from $\tilde{\x}$ by \emph{rounding up} every $x_{ij}$ to the nearest integer multiple of $(\varepsilon/nL)^{n^2+1}$. This is feasible as both the upper bound ($1$) and lower bound ($0$) are integer multiples of $(\varepsilon/nL)^{n^2+1}$. The algorithm runs in time $\textup{poly}(n,L,1/\varepsilon)$ according to Theorem~\ref{thm:localsearch}. We now claim that $\x' \in X$ and $P(\x') < P(\x)$.

\begin{lemma}
    \label{lem:CLS-feasibility}
    We have $\x' \in X$.
\end{lemma}

\begin{proof}
    Every $x'_{ij}$ is an integer multiple of $(\varepsilon/nL)^{n^2+1}$. Since $G(\x, \varepsilon/(nL))$ is acyclic, Algorithm~\ref{alg-LS} ensures that $G(\tilde{\x}, \varepsilon/(nL))$ is also acyclic. Since $x'_{ij} \geq \tilde{x_{ij}}$ for all $i,j$, the set of edges in $G(\x', \varepsilon/(nL))$ will be a subset of the set of edges in $G(\tilde{\x}, \varepsilon/(nL))$, implying that $G(\x', \varepsilon/(nL))$ is acyclic and thus $\x' \in X$.
\end{proof}

\begin{lemma}
    \label{lem:CLS-potential}
    We have $P(\x') < P(\x)$.
\end{lemma}

\begin{proof}
    Theorem~\ref{thm:localsearch} ensures that $P(\tilde{\x}) < P(\x) - L/2n^2$. Therefore, it suffices to show that $P(\x') \leq P(\tilde{\x}) + L/2n^2$. To this end, first observe that $|x'_{ij} - \tilde{x}_{ij}| \leq (\varepsilon/nL)^{n^2+1}$. Therefore, we have $|u_i(\x'_i) - u_i(\tilde{\x}_i)| \leq nL \cdot (\varepsilon/nL)^{n^2+1} \leq (\varepsilon/nL)^{n^2}$. Similarly, we have for any $i,j$, $|\psi_{ij}(\x') - \psi_{ij}(\tilde{\x})| \leq nL \cdot (\varepsilon/nL)^{n^2+1} \leq (\varepsilon/nL)^{n^2}$. This implies that $|\sum_j \psi_{ij}(\x') - \sum_j \psi_{ij}(\tilde{\x})| \leq n \cdot (\varepsilon/nL)^{n^2} \leq (\varepsilon/nL)^{n^2-1}$. With this, we can conclude that
    \begin{align}
       \label{surplus-diff}
        |\Delta_i(\x') - \Delta_i(\tilde{\x})| \leq 2(\varepsilon/nL)^{n^2-1} ~~\forall i
    \end{align}
    Equation~\ref{surplus-diff} implies that 
    \begin{align*}
        P(\x')- P(\x) &\leq \sum_{i} L^{2(n-i)} |\Delta_i(\x') - \Delta_i(\tilde{\x})|\\
                   &\leq \sum_i L^{2(n-i)} \cdot 2(\varepsilon/nL)^{n^2-1} &\text{(by Equation~\ref{surplus-diff})}\\
                   &\leq (\varepsilon/nL)^n\\
                   &<L/2n^2.
    \end{align*}
    Therefore, we have $P(\x') < P(\x)$.
\end{proof}

To summarize, we have a finite set of valid exchanges $X$. Each exchange $x \in X$, has a polynomial time computable cost function $P(x)$ associated with it. We define the neighborhood of every $x \in X$, and also give a polynomial time algorithm that given an exchange $x \in X$, outputs a neighbor $x'$ of $x$ such that $P(x') < P(x)$. Lastly, we show any local optima corresponds to a $\varepsilon$-core-stable and $\varepsilon$-reciprocal exchange. This shows that the problem is in PLS.

\begin{theorem}
    \label{thm:PLS-membership}
     Determining a $\varepsilon$-reciprocal and $\varepsilon$-core-stable exchange is in PLS, when $L/\varepsilon = \textup{poly}(n)$.
\end{theorem}

\subsection{Membership in PPAD}\label{app:ppad}
In this section, we show that the problem of computing an $\eps$-reciprocal and $\eps$-\cs exchange lies in PPAD under the assumption that utility and sharing functions are $L$-Lipschitz continuous for some rational  $L$ satisfying $L/\eps = \poly(n)$. We rely on the following theorem from \cite{etessami2010fixedpoint} which shows that the problem of computing an \textit{almost fixed point} of certain kinds of functions lies in PPAD.
\begin{proposition}[\cite{etessami2010fixedpoint}]\label{prop:ppad}
Let $\mathcal{F} = \{F_I\}$ be a family of functions $F_I$ associated with the instances $I$ of a search problem $\Pi$, such that $F_I$ is a continuous function that maps a convex, compact domain $D_I$ to itself, and the solutions of $I$ are the fixed points of $F_I$. The problem of computing an `almost fixed point' is: Given $I$ and rational $\eps>0$, compute a rational $x\in D_I$ such that $|F_I(x) - x| < \eps$.

\noindent Suppose $\mathcal{F}$ satisfies:
\begin{itemize}
\item[(i)] (Polynomial computability) There is a polynomial $q(\cdot)$ s.t. for all instances $I$, (a) the domain $D_I$ is a convex polytope described by a set of inequalities with rational coefficients that can be computed from $I$ in time $q(|I|)$, and (b) for any rational $x\in D_I$, $F_I(x)$ is rational and can be computed from $I$ and $x$ in time $q(|I| + size(x))$.
\item[(ii)] (Polynomial continuity) There is a polynomial $q(\cdot)$ s.t. for all instances $I$ and rational $\rho > 0$, there exists rational $\gamma > 0$ s.t. $size(\gamma) \le q(|I| + size(\rho))$ and for all $x,y\in D_I$, $y\in \mathcal{B}(x, \gamma)$ implies $F_I(y) \in \mathcal{B}(x, \rho)$. 
\end{itemize}
Then, the problem of computing an almost fixed point for $\mathcal{F}$ is in PPAD. 
\end{proposition}

Ideally, we would like to use the function $g: Z\rightarrow Z$ from our fixed point formulation from \cref{sec:existence}, and prove that the problem of computing an almost fixed point of $g$ is in PPAD. However, $g$ and $Z$ do not directly satisfy the conditions of \cref{prop:ppad} since the formulation suffers from a lack of polynomial computability: 
\begin{itemize}
\item[(i)] It is not clear if the constant $\alpha = \alpha(\eps)$ appearing in the constraints of $Z$ can be efficiently computed. 
\item[(ii)] Although the domain $Z$ is a convex polytope, it has exponentially many constraints, and hence cannot be directly computed from the instance in polynomial time.
\item[(iii)] The oracle computing the utility and sharing functions returns rational outputs when given a rational exchange $\x$ as input. However, for rational $z\in Z$, the exchange $\x = f(z)$ need not be rational, as $x_{ij} = f(z)_{ij} = 1 - \exp(-z_{ij})$.
\end{itemize}

We address these issues as follows.
\begin{itemize}
\item[(i)] We first re-define the set $Z$. Let $b \ge n\cdot \log(\frac{nL}{\eps})$ be a rational number, and let $M = n\cdot b$. Let $Z$ be the set given by $Z = \{z\in[0, M]^{n\times n} : \forall \text{ cycles } C \in K_n, \sum_{e\in C} z_e \ge b\}$. By choice of $M$ and $b$, the constants appearing in the constraints of $Z$ are rational numbers computable in polynomial time and representable using $O(\log n + \log(\frac{nL}{\eps}))$ bits. 

Further, by definition of $Z$ and choice of $b$, for any $z\in Z$, the exchange graph $G(f(z), \frac{\eps}{nL})$ is acyclic. \cref{lem:stronger-assumptions} then implies that the exchange $f(z)$ is $\eps$-\cs for any $z\in Z$. However, $Z$ still has exponentially many constraints. 
\item[(ii)] Therefore, we embed $Z$ in the larger polytope $Z' = [0, M]^{n\times n}$, since $Z\subseteq Z'$ by definition. We define a new function $\tilde{g}: Z'\rightarrow Z'$ whose fixed points are exactly the fixed points of $g$. For this, let $\pi : Z'\rightarrow Z$ be the \textit{projection} of $Z'$ onto $Z$, i.e., $\pi$ is the identity on $Z$ and maps points in $Z'\setminus Z$ to their nearest point on the boundary of $Z$. Thus, $\pi(z') = \argmin_{z \in Z} \elltwo{z'-z}$ for $z'\in Z'$. Note that $\pi$ is a well-defined function since the set $Z$ is closed and convex, and the $L_2$ norm is strictly convex. The function $\tilde{g}$ is the composition of $\pi$ and $g$, i.e. $\tilde{g}(z) = g(\pi(z))$ for $z\in Z'$. It is easy to see that any $z\in Z'\setminus Z$ cannot be a fixed point of $\tilde{g}$, and hence the set of fixed points of $g$ and $\tilde{g}$ are identical. 
\item[(iii)] Finally we address the polynomial computability of $g(z)$ for a rational $z\in Z$, given that the utility and sharing oracles only return rational outputs when given rational exchanges $\x$ as input. For rational $z\in Z$, $f(z)$ need not be rational as $f(z)_{ij} = 1 - \exp(-z_{ij})$. Therefore, we discuss how to obtain a \textit{continuous, rational approximation} $\hat{f}(z)$ of $f(z)$ through standard techniques \cite{etessami2010fixedpoint}. For some small parameter $\kappa$, we partition $[0,M]^{n\times n}$ into cubelets of length $\kappa$ along cardinal directions, and then construct a simplicization of each cubelet in polynomial time. Given a rational $z\in Z$, we compute $\hat{f}(z)$ in polynomial time as follows: identify the simplex to which $z$ belongs by binary search, and uniquely express $z$ as a convex combination of the $n^2+1$ vertices $\{z_i\}_{i=1}^{n^2+1}$ of the simplex, i.e., $z = \sum_{i=1}^{n^2+1} \lambda_i z_i$. Compute $\hat{f}(z_i)$ to precision $\kappa$, and then compute $\hat{f}(z)$ by linearly interpolation, that is, $\hat{f}(z) = \sum_i \lambda_i \hat{f}(z_i)$. In this manner, we obtain a continuous, rational approximation function $\hat{f}$ for $f$ up to an error tolerance of $\kappa$, i.e., $|\hat{f}(z) - f(z)| \le O(n^2\cdot\kappa)$. Using the rational exchange $\hat{f}(z)$ computed from a rational $z\in Z$, we can then approximately compute $g(z)$ (and hence also $\tilde{g}(z')$ for rational $z'\in Z'$) within a $\poly(\kappa, n)$ error. It is then straightforward to see that an $\eps$-almost fixed point of the rational approximation of $\tilde{g}$ is an $(\eps + \poly(\kappa, n))$-almost fixed point of $\tilde{g}$. Thus, with the choice of $\kappa = \eps/\poly(n)$, the construction of the continuous rational approximation $\tilde{f}$ in polynomial time effectively circumvents the issue of irrational $f(z)$ for rational $z$.

In what follows, we assume for simplicity that we can compute a rational value of $f(z)$ from rational $z\in Z$. 
\end{itemize}

We are now in a position to prove that $\tilde{g} : Z \rightarrow Z'$ is polynomially computable, i.e., satisfies condition (i) of \cref{prop:ppad}.

\begin{lemma}\label{lem:PPAD-computation}
The function $\tilde{g}: Z' \rightarrow Z'$ is polynomially computable.
\end{lemma}
\begin{proof}
Observe that $Z' = [0, M]^{n\times n}$ is a convex polytope described by a set of $O(n^2)$ inequalities with rational coefficients, since $M$ is rational and polynomial time computable by choice of $M$. Thus the domain $Z'$ satisfies property (a) of \cref{prop:ppad} (i).

Next, we prove property (b) of \cref{prop:ppad} (i) is satisfied as well: we show that for a given rational $z'\in Z'$, $\tilde{g}(z')$ is rational and can be computed in time $\poly(size(L), size(M), size(z'))$. Since $\tilde{g} = g\circ \pi$, it suffices to separately prove this for both $\pi$ and $g$. 

For $g$, it suffices to prove the property for each $g^{ij}$. Since we are given access to oracles computing the utility and share functions, we only need to argue that $\beta^+_{ij}(z)$ and $\beta^{-}_{ij}(z)$ are rational and polynomial time computable for rational $z\in Z$. Note that by \cref{eq:set-z}, $Z$ is a polytope and can be expressed as $Z = \{z : Az \le c\}$ for some $A\in \{-1, 0\}^{R\times n^2}$ and $c\in \mathbb{Q}^R$, where $A$ has $R$ constraints (exponentially many in $n$) and $n^2$ variables. 

Let us first note that $Z$ admits a polynomial time separation oracle: given $z\in [0,M]^{n\times n}$, we (i) compute the exchange $f(z) \in [0, 1]^{n\times n}$, (ii) construct the graph $G(f(z), \frac{\eps}{nL})$, and (iii) identify a cycle $C$ in $G(f(z), \frac{\eps}{nL})$. It is clear that $G(f(z), \frac{\eps}{nL})$ is acyclic iff $z\in Z$; and if $G(f(z), \frac{\eps}{nL})$ has a cycle $C$, then the constraint $\sum_{e\in C} z_e < b$ is a hyperplane separating $z$ from $Z$. Further each of these steps can be performed in polynomial time.

With the efficient separation oracle for $Z$, we show how to compute $\beta^+_{ij}(z)$ for $z\in Z$. Observe that $\beta^+_{ij}$ is the solution to the linear program given by $\max \{b \ge 0 : A\cdot (z + b\cdot \e_{ij}) \le c\}$. Since $A, z$ and $c$ are rational, the LP is one-dimensional, and the constraint set admits an efficient separation oracle, the optimal solution to the LP $\beta^+_{ij}(z)$ is rational and can be computed in $\poly(size(M), size(L), size(z))$ for any $z\in Z$. Likewise, we can show that $\beta^-_{ij}(z)$ is polynomially computable. This shows that the function $g$ is polynomially computable.

Next, we argue that the function $\pi$ is polynomially computable. Since $\pi(z) = z$ for $z\in Z$, we only need to argue that $\pi(z')$ is rational and can be computed in polynomial time for rational $z'\in Z'\setminus Z$. Recall that $\pi(z')$ is the optimal solution to the following quadratic program $\pi(z') = \arg\min_{z\in Z} \elltwo{z'-z}^2$. Since we have an efficient separation oracle for the constraint set, that the quadratic program can be solved in polynomial time using the ellipsoid method. Lastly, we argue that for a rational $z'\in Z'\setminus Z$, the optimum solution $z = \pi(z')$ is rational and $size(z) = \poly(size(L), size(M), size(z'))$.

The quadratic program $z \in \arg\min_{z\in Z} \elltwo{z'-z}^2$ can be written as:
\begin{align*}
&\quad\min \frac{1}{2}\sum_{i, j} (z_{ij} - z'_{ij})^2 \\
& \sum_{(i,j)\in C} z_{ij} \ge b, \text{ for all cyclic permutations } C = (i_1, \dots, i_k) \text{ of } [n] \\
& 0 \le z_{ij} \le M \text{ for all } i, j\in [n]
\end{align*}
 
Let $y_C$ denote the dual variable for the constraint corresponding to cycle $C = (i_1, \dots, i_k)$ and let $\mu_{ij}$ denote the dual variable for the $z_{ij} \le M$ constraint. Expanding the objective and writing the KKT conditions shows us that the optimal solution $z$ and the dual solutions $y_C, \mu_{ij}$ are solutions to the following linear complementarity program (LCP):

\begin{subequations}\label{lcp:quad}
\begin{eqnarray}
\forall i, j \in [n]: & \ z_{ij} - z'_{ij} \le u_{ij} + \sum_{C: (i,j)\in C} y_C  & \perp z_{ij} \label{eq:lcp-1}\\
\forall i,j \in [n]: & \ z_{ij} \le M & \perp \mu_{ij} \label{eq:lcp-2} \\
\forall C: & \ \sum_{(i,j)\in C} z_{ij} \ge b & \perp y_C \label{eq:lcp-3}
\end{eqnarray}
\end{subequations}

Note $z$ is either a vertex or lies on the facet of the $n^2$-dimensional polytope $Z$. Hence, at most $n^2$ constraints of $Z$ can be tight at $z$. Due to the complementarity slackness condition \eqref{eq:lcp-3}, we see that $y_C = 0$ for all but $n^2$ variables $y_C$. Thus, the solution $(z, y_C, \mu_{ij})$ has at most $3 n^2$ non-zero entries. Let $I$ be the set of variables taking non-zero values in $(z, y_C, \mu_{ij})$. By the complementarity condition of the LCP, the inequalities associated with $I$ hold with equality, and form a linear system $A'\cdot x = c'$, where $A'$ is a $|I|\times |I|$ square sub-matrix, and $x, c$ are an $|I|$ dimensional vectors with $x$ comprising of the variables in $I$ with non-zero values. This allows us to express $x = (A')^{-1}\cdot c'$. Since the entries of $A'$ and $c'$ are rationals of size $\poly(\log M)$ and the dimension of $A'$ is $3n^2$, we can conclude that the entries of $x$ will be rational numbers of size $\poly(n, \log M)$. In particular, for all $z'\in Z'$, $z = \pi(z')$ is a rational number of size $\poly(n+\log M)$. Thus, $\pi$ is polynomially computable. 
\end{proof}

\begin{lemma}\label{lem:PPAD-continuous}
The function $\tilde{g}: Z'\rightarrow Z'$ is polynomially continuous.
\end{lemma}
\begin{proof}
We separately prove that both $g$ and $\pi$ are polynomially continuous; the polynomial continuity of $\tilde{g} = g\circ \pi$ follows from the fact that function composition preserves polynomial continuity.

First, $\pi$ is polynomially continuous since it is Lipschitz continuous with Lipschitz constant $1$. This follows from the well-known fact that the projection operator is non-expansive for convex sets (see Proposition 2.2.1 of \cite{bertsekas2003convex}).

Next, we show that $g: Z\rightarrow Z$ is polynomially continuous. We revisit the proof of \cref{lem:continuous} showing the continuity of $g$ under the assumption that the utility and sharing functions are $L$-Lipschitz continuous. As before, it suffices to prove that for fixed $i,j\in[n]$, $g^{ij}$ is polynomially continuous. Fix $z\in Z$ and a rational $\rho > 0$. We show that exists a rational $\gamma > 0$ with $size(\gamma) = \poly(size(\rho), size(M), size(L))$ s.t. $z' \in \mathcal{B}(z, \gamma)$ implies $g^{ij}(z') \in \mathcal{B}(g^{ij}(z), \rho)$. Recall $h^{ij}(z) = \Delta_j(f(z)) - \Delta_i(f(z))$. We consider three cases:
\begin{itemize}
\item $h^{ij}(z) = 0$. Since $h^{ij}$ is Lipschitz continuous, there exists some rational $\gamma' > 0$ such that $size(\gamma') = \poly(size(L))$ and $z'\in \mathcal{B}(z, \gamma')$ implies $h^{ij}(z') \in \mathcal{B}(h^{ij}(z), \frac{\rho}{2\delta M}) = \mathcal{B}(0, \frac{\rho}{2\delta M})$. For $z'\in \mathcal{B}(z, \gamma')$, we define $\beta_{ij}(z') = \beta^+_{ij}(z')$ if $h^{ij}(z') \ge 0$, and $\beta^-_{ij}(z')$ otherwise. Using the definition of $g^{ij}$ from \cref{eq:g-ij}, we have for every $z'\in\mathcal{B}(z, \gamma')$:
\begin{align*}
|g^{ij}(z') - g^{ij}(z)| &\le |z'-z| + \delta\cdot|\beta_{ij}(z')\cdot h^{ij}(z')| \\
&\le |z'-z| + \delta\cdot M \cdot \frac{\rho}{2\delta M} \tag{using $|\beta^{ij}(z')| \le M$}\\
&\le |z'-z| + \frac{\rho}{2}.
\end{align*}
We set $\gamma := \min(\gamma', \rho/2)$. Then, the above inequality implies that for all $z'\in \mathcal{B}(z, \gamma)$, we have $g^{ij}(z') \in \mathcal{B}(g^{ij}(z), \rho)$. Moreover, $\gamma$ is rational and $size(\gamma) = \poly(size(L), size(\rho))$ since $size(\gamma') = \poly(size(L))$. This shows that $g^{ij}$ is polynomially continuous at $z$.

\item $h^{ij}(z) = r > 0$. Since $h^{ij}$ is Lipshcitz continuous, there exists some $\gamma' > 0$ such that $size(\gamma') = \poly(size(L))$ and $z' \in \mathcal{B}(z, \gamma)$ implies $h^{ij}(z') \in \mathcal{B}(h^{ij}(z), \frac{r}{2})$. Thus, $|h^{ij}(z') - h^{ij}(z)|  < r/2$. Since $h^{ij}(z) = r$, we conclude $h^{ij}(z') > 0$ for all $z'\in \mathcal{B}(z, \gamma')$. This implies that $g^{ij}(z') = z' + \e_{ij}\cdot \delta\cdot\beta^+_{ij}(z')\cdot h^{ij}(z')$ for all $z'\in \mathcal{B}(z, \gamma')$. 
Now consider the function $\bar{g}^{ij} : Z \rightarrow Z$ given by $\bar{g}^{ij}(z') := z' + \e_{ij}\cdot\delta\cdot\beta^+_{ij}(z')\cdot h^{ij}(z')$. We have that $g^{ij}(z') = \bar{g}^{ij}(z')$ for all $z'\in \mathcal{B}(z, \gamma')$.

First we observe that $\beta^+_{ij}$ is polynomially continuous. This is because we showed that $\beta^+_{ij}$ is concave in \cref{lem:convex-func}, and hence it is Lipschitz continuous with Lipschitz constant $\poly(M)$. Since $\beta^+_{ij}$ and $h^{ij}$ are both polynomially continuous functions, $\bar{g}^{ij}$ is also a polynomially continuous function. Thus, there exists some rational $\gamma'' > 0$ s.t $size(\gamma'') = \poly(size(M), size(L))$ and $z'\in\mathcal{B}(z, \gamma'')$ implies $\bar{g}^{ij}(z') \in \mathcal{B}(\bar{g}^{ij}(z), \rho)$. 

We set $\gamma := \min(\gamma', \gamma'')$. We obtain that $z' \in \mathcal{B}(z, \gamma)$ implies $g^{ij}(z') \in \mathcal{B}(g^{ij}(z), \rho)$, since $g^{ij}(z') = \bar{g}^{ij}(z')$ for all $z'\in \mathcal{B}(z, \gamma')$. Moreover, $\gamma$ is rational and $size(\gamma) = \poly(size(L), size(M))$. This proves that $g^{ij}$ is polynomially continuous at $z$.
\item $h^{ij}(z) = r < 0$. An analysis analogous to the above case shows that $g^{ij}$ is polynomially continuous at $z$ in this case too.
\end{itemize}
In conclusion, $g$ is polynomially continuous. Hence $\tilde{g} = g\circ \pi$ is also polynomially continuous.
\end{proof}

With the above two lemmas, we can apply \cref{prop:ppad} implies:
\begin{lemma}\label{lem:g-tilde-ppad}
The problem of computing an almost fixed point of the function $\tilde{g}:Z'\rightarrow Z'$ is in \textup{PPAD}.
\end{lemma}

We now use \cref{lem:g-tilde-ppad} to prove the PPAD-membership of computing an approximately-reciprocal and approximately-\cs exchange.
\begin{theorem}\label{thm:ppad}
For every $\eps \in (0,1)$, the problem of computing an $\eps$-reciprocal and $\eps$-\cs exchange is in \textup{PPAD}.
\end{theorem}
\begin{proof}
We reduce the problem of computing an $\eps$-reciprocal and $\eps$-\cs exchange to the problem of computing an almost fixed point of $\tilde{g}$, i.e., a point $z'\in Z'$ s.t. $\elltwo{\tilde{g}(z') - z'} \le \frac{\eps^2}{2n^8 L^3}$. The theorem follows since the latter problem is in PPAD by \cref{lem:g-tilde-ppad}.

Let $\eps' = \frac{\eps^2}{2n^7 L^2}$. The reduction maps an almost fixed point $z'\in Z'$ satisfying $\elltwo{\tilde{g}(z') - z'} \le \eps'$ to the exchange $\x := f(\pi(z'))$. Clearly, $\x$ can be computed from $z'$ in time $\poly(n, \frac{1}{\eps}, L)$. We prove that $\x$ is $\eps$-reciprocal and $\eps$-\cs. 

First we show that the point $z = \pi(z') \in Z$ satisfies $\elltwo{g(z) - z)} \le 2\eps'$. Observe that $\tilde{g}(z') = g(z)$ since $\tilde{g}(z') = g(\pi(z'))$. Further, since $z$ is the point on $Z$ closest to $z'$ and $g(z)$ is another point in $Z$, we have $\elltwo{z'-z} \le \elltwo{z'-g(z)}$. Thus:
\begin{align*}
\elltwo{g(z)-z} &\leq \elltwo{g(z)-z'} + \elltwo{z'-z}\\
&\le 2\cdot \elltwo{g(z) - z} \tag{since $\elltwo{z'-z} \le \elltwo{z'-g(z)}$} \\
&= 2\cdot \elltwo{\tilde{g}(z') - z'} \tag{since $\tilde{g}(z') = g(z)$}\\
&\le 2\cdot \eps'.
\end{align*}

Next we consider the exchange $\x = f(z)$ where $x_{ij} = 1 - \exp(-z_{ij})$. Since $z\in Z$, \cref{lem:stronger-assumptions} and the definition of $Z$ imply that $\x$ is $\eps$-\cs. We claim that:
\begin{claim}
\label{obs:approxiation-PPAD}
If $\x$ is not $\varepsilon$-reciprocal, then $\elltwo{g(z) - z} > \frac{\varepsilon^2}{n^7 L^2}$.
\end{claim}
However, we showed earlier that $\elltwo{g(z) - z} \le 2\cdot \eps' = 2\cdot \frac{\eps^2}{2n^7 L^2} = \frac{\eps^2}{n^7 L^2}$. Thus \cref{obs:approxiation-PPAD} implies that $\x$ must be $\eps$-reciprocal. We now prove the claim.

\begin{proof}[Proof of \cref{obs:approxiation-PPAD}] 
If $\x$ is not $\varepsilon$-reciprocal, then there exists an agent $i$ such that  $\Delta_i(\x) > \varepsilon/n$. We pick the set of high surplus agents  $S$ by Algorithm~\ref{alg-S}. Note that we have $\min_{i \in S} \Delta_i(\x) > \max_{i \in \N \setminus S} \Delta_i(\x) + \varepsilon/n^2$. By Lemma~\ref{obs:non-zeroflow}, we have an agent $j \in \N \setminus S$, such that $\sum_{i \in S} \psi_{ij}(\x) \geq \varepsilon/n^2$, implying that there exists an $i \in S$, such that $\psi_{ij}(\x) \geq \varepsilon/n^3$. The Lipschitzness of the utility functions imply that $x_{ij} \geq \varepsilon/(n^3L)$, showing that:
\begin{align*}
z_{ij} = \log(\frac{1}{1-x_{ij}}) \geq \log(\frac{1}{1-\varepsilon/(n^3L)}) = -\log(1-\frac{\varepsilon}{n^3L}) \geq \frac{\varepsilon}{n^3L},
\end{align*}
where last inequality used the fact that $\log(1+y) \leq y$ for all $y > -1$. We now distinguish two cases, depending on the value of $\beta^{-}_{ij}(z)$.

\paragraph{Case 1: $\beta^{-}_{ij}(z) \geq \varepsilon/(n^3L)$.} First observe that from the definition of  $g(\cdot)$, we have,
\begin{align}\label{eq:g}
(g(z))_{ij} - z_{ij} = \frac{1}{n^2} \cdot (g^{ij}(z)-z)_{ij} 
\end{align}

Since $\Delta_i(f(z))- \Delta_j(f(z)) = \Delta_i(x)- \Delta_j(x) > \varepsilon/n^3 > 0$, we have  $(g^{ij}(z) - z)_{ij}= \delta \cdot \beta^{-}_{ij} (\Delta_i(f(z))- \Delta_j(f(z))$. Since $\beta^{-}_{ij}(z) \geq \varepsilon/(n^3L)$, and $\Delta_i(f(z))-\Delta_j(f(z)) = \Delta_i(x)- \Delta_j(x) > \varepsilon/n^3$, and $\delta > 1/nL$ from \cref{lem:stronger-assumptions}, we have
\begin{align*}
|(g(z))_{ij} - z_{ij}| \geq \varepsilon^2/(n^7 L^2),
\end{align*}
implying that $\elltwo{g(z) - z} > \varepsilon^2/(n^7 L^2)$.

\paragraph{Case 2: $\beta^{-}_{ij}(z) < \varepsilon/(n^3L)$:}
We have already established $z_{ij} \geq \varepsilon/n^3L$, meaning that if we reduce $z_{ij}$ by $\varepsilon/(n^3L)$, then the constraint $z_{ij} \geq 0$ will not be violated. Then, the only constraint that can be violated is some constraint of the form $\sum_{\ell \in [k]} z_{i_{\ell} i_{\ell+1}} \geq n \log(nL/\varepsilon)$ (indices are modulo $k$), where $i_1, i_2, \dots, i_k$ is some sequence of agents with $i_1 = i$ and $i_2 = j$. 

This implies that  $ \sum_{\ell \in [k]} z_{i_{\ell} i_{\ell+1}} < n \log(nL/ \varepsilon) + \varepsilon/(n^3L)$, trivially implying that each $z_{i_{\ell} i_{\ell+1}} < n \log(nL/ \varepsilon) + \varepsilon/(n^3L)$ for $\ell \in [k]$. Observe that since $0 \leq z_{i_\ell i_{\ell+1}} \leq M$, there is still some slack in increasing $z_{i_\ell i_\ell}$. In particular, this slack is 
\begin{align*}
M - n \log(nL/ \varepsilon) - \varepsilon/(n^3L) &> n^2\cdot\log (nL)/\varepsilon)- n \log(nL/ \varepsilon) - \varepsilon/(n^3L)\\ 
                  &> n \log(nL/\varepsilon) 
\end{align*}
as $M = n^2\log(nL/\varepsilon)$.  This suggests that $\beta^{+}_{i_{\ell} i_{\ell+1}}(z) \geq n \log(nL/\varepsilon)$ for all $\ell \in [k]$, since no constraint of $Z$ gets violated on increasing $z_{ij}$ other than $z_{ij} \le M$. Let $\mathit{gap}_{\ell}= \Delta_{i_{\ell}}(f(z))-\Delta_{i_{\ell+1}}(f(z))$ for all $\ell \in [k]$. Note that, we have $\sum_{\ell =2}^{k} \mathit{gap}_{\ell} =  \Delta_{j}(f(z)) - \Delta_i(f(z)) < -\varepsilon/n^2$, as it is a telescoping sum and $i_2 = j$ and $i_1=i$. This implies that there exists an $\ell \in [k] \setminus \{1\}$, where $\mathit{gap}_{\ell} < -\varepsilon/n^3$ as $k \leq  n$. Therefore, there exists agents  $i_{\ell}$ and $i_{\ell+1}$ such that $\Delta_{i_{\ell}}(f(z)) - \Delta_{i_{\ell+1}}(f(z)) < -\varepsilon/n^3$, and $\beta^{+}_{i_{\ell} i_{\ell+1}}(z) \geq n \log(nL/\varepsilon)$. Since $\Delta_{i_{\ell}}(f(z)) - \Delta_{i_{\ell+1}}(f(z)) < -\varepsilon/n^3 < 0$, we have:
\begin{align*}
|(g^{i_{\ell}i_{\ell+1}}(z) - z)_{i_{\ell}i_{\ell+1}}| &= |\delta \cdot \beta^{+}_{i_{\ell}i_{\ell+1}}(z)\cdot (\Delta_{i_{\ell+1}}(f(z))- \Delta_{i_{\ell}}(f(z)))|\\
& > \frac{1}{n L} \cdot (n \log(nL/\varepsilon)) \cdot \frac{\varepsilon}{n^3}\\
&\geq\frac{\varepsilon^2}{n^4 L^2}.
\end{align*}
Equation~\ref{eq:g} then implies that $||g(z) - z||_2 > \frac{1}{n^2} \cdot \frac{\varepsilon^2}{n^4 L^2} \geq \frac{\varepsilon^2}{n^6 L^2} \geq \frac{\varepsilon^2}{n^7 L^2}$.
\end{proof}
In summary, the exchange $\x$ is $\eps$-reciprocal and $\eps$-\cs exchange that can be computed in $\poly(n, \frac{1}{\eps}, L)$ time from an $(\frac{\eps^2}{n^7 L^2})$-almost fixed point of $\tilde{g}$.
\end{proof}

\paragraph{Membership in CLS} \cref{thm:PLS-membership} shows that determining a $\varepsilon$-core stable and $\varepsilon$-reciprocal exchange is in PLS. \cref{thm:ppad} shows that the same problem is in PPAD. Together, we obtain that the problem lies in PPAD $\cap$ PLS $=$ CLS, thus proving \cref{thm:CLS}.

\end{document}